\documentclass[twoside,11pt]{article}
\usepackage{mathtools}
\usepackage{amsmath}
\usepackage{amssymb}
\usepackage{graphicx}
\usepackage{comment}
\usepackage{xcolor}
\usepackage{bm}
\usepackage{amsthm}
\usepackage{siunitx}
\usepackage{multirow}
\usepackage{algorithm}
\usepackage{algorithmicx}
\usepackage{algpseudocode}
\algrenewcommand\algorithmicrequire{\textbf{Input:}}
\algrenewcommand\algorithmicensure{\textbf{Output:}}
\newtheorem{assumption}{Assumption}
 
\newtheorem{theorem}{Theorem}
\newtheorem{lemma}[theorem]{Lemma} 
\newtheorem{proposition}[theorem]{Proposition} 
\theoremstyle{remark}
\newtheorem*{remark}{Remark}
\newtheorem{corollary}[theorem]{Corollary}

\usepackage{caption}
\usepackage{subcaption}
\def\ci{\perp\!\!\!\perp}
\usepackage{enumerate}
\usepackage{natbib}
\usepackage{url}

\usepackage{jmlr2e}

\usepackage{lastpage}

\ShortHeadings{Identifiable Deep Latent Variable Models for MNAR Data}{Xie, Xue, Wang}
\firstpageno{1}

\begin{document}

\title{Identifiable Deep Latent Variable Models for MNAR Data}

\author{\name Huiming Xie \email xie339@purdue.edu \\
       \addr Department of Statistics\\
       Purdue University\\
       West Lafayette, IN 47907, USA
       \AND
       \name Fei Xue \email feixue@purdue.edu \\
       \addr Department of Statistics\\
       Purdue University\\
       West Lafayette, IN 47907, USA
       \AND
       \name Xiao Wang \email wangxiao@purdue.edu \\
       \addr Department of Statistics\\
       Purdue University\\
       West Lafayette, IN 47907, USA}

\editor{}

\maketitle

{
\begin{abstract}
Missing data is a ubiquitous challenge in data analysis, often leading to biased and inaccurate results. Traditional imputation methods usually assume that the missingness mechanism is missing-at-random (MAR), where the missingness is independent of the missing values themselves. This assumption is frequently violated in real-world scenarios, prompted by recent advances in imputation methods using deep learning to address this challenge. However, these methods neglect the crucial issue of nonparametric identifiability in missing-not-at-random (MNAR) data, which can lead to biased and unreliable results. This paper seeks to bridge this gap by proposing a novel framework based on deep latent variable models for {MNAR data}. Building on the assumption of conditional no self-censoring {given} latent variables, we establish the identifiability of the data distribution. This crucial theoretical result guarantees the feasibility of our approach. To effectively estimate unknown parameters, we develop an efficient algorithm utilizing importance-weighted autoencoders. We demonstrate, both theoretically and empirically, that our estimation process accurately recovers the ground-truth joint distribution under specific regularity conditions. Extensive simulation studies and real-world data experiments showcase the advantages of our proposed method compared to various classical and state-of-the-art approaches to missing data imputation.
\end{abstract}

\section{Introduction}\label{intro}
The issue of missing data poses a pervasive challenge across diverse applications, often necessitating data imputation due to the prevalence of models designed exclusively for complete data. 
Methods for addressing missing data have been studied extensively in ignorable missing situations where missingness depends only on the fully observed variables of the data, known as missing-at-random (MAR), or is independent of any data variable, known as missing-completely-at-random 
 (MCAR) \citep{robins1994estimation}. 
%
On the basis of such ignorability assumptions, many interesting imputation methods have been developed to make use of most models designed for complete data.
For example, missForest \citep{stekhoven2012missforest} based on random forests is one of the popular methods for single imputation, while MICE \citep{van2011mice}, short for multivariate imputation by chained equations, is a typical method for multiple imputation \citep{little2019statistical}.
%
Supported by increased computing power, recent advances have witnessed the development of sophisticated imputation methods, particularly deep generative models, based on these ignorability assumptions \citep{yoon2018gain,ma2018eddi,li2019misgan,mattei2019miwae}.

While convenient for modeling, these ignorability assumptions often do not align with the reality of real-world applications, spanning domains such as surveys, recommendation systems, and the medical field. For instance, in surveys, 
a question associated with social stigma 
   could  result in nonresponse \citep{marra2017simultaneous, malinsky2021semiparametric}. 
Certain medical records could be incomplete 
either because the information is unnecessary or because acquiring it is challenging due to the patients' health conditions \citep{shpitser2016consistent, jakobsen2017and}.  
In situations where absence of data is potentially influenced by characteristics of missing variables, known as missing-not-at-random (MNAR), modeling becomes significantly more complex \citep{robins1997non}.

There has been a relatively limited body of work addressing MNAR settings within the context of scalable missing value imputation. The scenario is explored more often in the application of recommender systems using probabilistic models \citep{hernandez2014probabilistic, ling2012response, liang2016modeling}. Other perspectives include causal approaches with an explicit model of exposure for MNAR \citep{liang2016causal,wang2018deconfounded, wang2019blessings}, and inverse probability weighting methods \citep{ma2019missing, schnabel2016recommendations,wang2019doubly}. 

More recently, deep generative models, especially deep latent variable models (DLVMs) have been studied under MNAR assumptions. For example, not-MIWAE \citep{ipsen2020not}, short for the not-missing-at-random importance-weighted autoencoder, is a powerful model that incorporates prior knowledge of the missingness mechanism into density estimation. Nonetheless, 
these models often overlook the critical issue of nonparametric identifiability in MNAR data, as discussed in the seminal work by
\citet{robins1997non}. Some initial efforts in DLVMs to address identification issues, such as the deep generative imputation model (GINA) \citep{ma2021identifiable}, require the presence of auxiliary variables, which may not always be available. Consequently, these models tend to learn multiple non-unique complete data distributions that match the observed data \citep{robins1997non}. Overlooking this fundamental aspect of MNAR data in modeling can introduce bias and inconsistency, potentially impairing inference and downstream tasks.


On the other hand, numerous theoretical studies have delved into the identification of MNAR data using nonparametric or semiparametric inference techniques.
They often impose constraints on the missingness mechanism 
to uniquely identify a functional or parameter of interest. Examples include the group permutation missing process \citep{robins1997non}, sequential identification \citep{sadinle2018sequential}, discrete choice models \citep{tchetgen2018discrete}, nearest identified pattern restriction \citep{linero2017bayesian}, conditions in graphical models \citep{fay1986causal, ma2003identification, mohan2021graphical, nabi2020full}, the no self-censoring model \citep{malinsky2021semiparametric, shpitser2016consistent, sadinle2017itemwise}, and the self-censoring model with completeness assumptions \citep{li2022self}. 
However, these methods generally
impose a specific parametric model on the ground-truth distributions, which might be misspecified.
}

In this paper, we propose an {\it {i}dentifiable {m}issing-not-at-random {i}mportance-{w}eighted {a}uto{e}ncoder} (IM-IWAE) method to estimate the true distribution of all the variables and impute missing values for MNAR data.
The proposed method is built upon a new
{conditional} no self-censoring assumption \textcolor{black}{given} latent variables for missingness mechanisms, which assumes that the missing indicator of a random variable is conditionally independent of this variable and other missing indicators given all other variables and latent factors.
%

Our contributions are outlined as follows: 
{\color{black}
\begin{itemize}
\item We establish the nonparametric identifiability of the true data distribution under our proposed conditional no self-censoring assumption for MNAR data. Beyond providing identifiability, this assumption clarifies how no self-censoring can be formulated within latent variable models, offering a foundation for subsequent methodological developments.



    \item We propose IM-IWAE, a working parametric model based on deep neural networks, which can asymptotically approximate the full-data distribution 
    without requiring parametric assumptions on the true data distribution.
    The model can be used to impute missing values and generate synthetic data from the underlying distribution, where the parameters are estimated through an algorithm based on importance weighted variational inference. 
    To our knowledge, this is the first approach that jointly addresses nonparametric identifiability of MNAR data and practical modeling through deep latent variable methods. 

    \item Through extensive numerical studies, we demonstrate the ability of our method to provide identifiable 
    estimates and valid imputations and generations under correct model assumptions, as well as its robustness to practical violations of the no self-censoring assumption in MNAR data. This empirical evaluation, conducted across simulations and real-world datasets,  shows advantages over classical imputation methods and  existing deep learning imputation methods in handling MNAR data. Reproducible Python code 
    is available at \url{https://github.com/hxstat/IM-IWAE}.
\end{itemize}
}

{
The remainder of the paper is organized as follows. Section \ref{gt_full} introduces the ground-truth model for the full-data distribution. Section \ref{gt_model} provides identification assumptions on the missingness mechanism in the ground-truth model. In Section \ref{sec: method}, we propose a deep latent variable approach
for learning the ground-truth data distribution in the presence of MNAR data. 
In Section \ref{theory}, we establish the theoretical properties of the proposed method. Sections \ref{simulationstudy} and \ref{realexp}
provide numerical studies through simulations and real data analysis, respectively. The paper ends with the conclusion of Section \ref{conclusion}.

\section{Ground-truth model of full-data distribution}\label{gt_full}
In this section, we present the ground-truth model for the distribution of all variables and missing indicators.
We begin by introducing the notation.
We use capital letters to denote random variables or vectors, and lowercases for their values. Consider a random vector $X = (X_1,\cdots,X_p)^T \in \mathcal{X}$,  where $\mathcal{X}$ is a $d$-dimensional  manifold embedded in $\mathbb{R}^p$ for $d \leq p$. Let $R = (R_1,\cdots,R_p)^T \in \{0,1\}^p$ be a vector of missingness indicators of $X$, 
where $1$ corresponds to ``observed", and $0$ corresponds to ``missing''. Let $(r^{(1)}, x^{(1)}), \cdots, (r^{(n)}, x^{(n)})$ denote $n$ independently and identically distributed (i.i.d.) samples of the pair $(R,X)$.
With a slight abuse of notation, we write $R=r$ as a shorthand for $(R_1,..,R_p) = (r_1, \cdots, r_p)$. Let $X_{(r)}$ be a subvector of observed elements in $X$  when
$R = r$. The observed data consist of i.i.d. realizations of the vector $(R, X_{(R)})$. 
For convenience, we define another indicator $M=1-R$ as the opposite of $R$, that is,  $M_j$ = 1 represents $X_j$  missing and 0 represents $X_j$ observed. 

For simplicity and clarity, we call the joint distribution of $(X, R)$ the {\it full-data distribution}. Let $p_{gt}({x},{r})=p_{gt}(x) p_{gt}(r|x)$ denote the {\it ground-truth full-data distribution}, where  $p_{gt}(x)$ is the {\it ground-truth complete data distribution of $X$}, and $p_{gt}(r|x)$ is the {\it ground-truth conditional distribution of $R$ given $X$}, known as  \emph{the missingness mechanism}.


When data are high-dimensional, latent variable models (LVMs) are powerful tools for characterizing and discovering low-dimensional structures in data \citep{kingma2013auto, chen2016infogan}. From this inspiration, we adopt an LVM for the data distribution $p_{gt}(x)$ by assuming that $X$ are generated from a latent variable $Z \in \mathbb{R}^{q_1}$ through the conditional distribution $p_{X|Z}(x|z)$ so that $p_{gt}(x) = \int p_{X|Z}(x|z) p_Z(z)dz$. Analogously, for the missingness mechanism $p_{gt}(r|x)$, we adopt another latent variable $\tilde{Z} \in \mathbb R^{q_2}$ to capture the information in $R$ that could not be explained by $X$ alone and assume that $R$ is generated through $p_{R|X, \tilde Z}(r|x, \tilde{z})$. For this reason, it is natural to have the following assumptions:
\begin{assumption}\label{p_gtx}
The ground-truth data distribution can be decomposed as
\begin{equation}\label{equ:gt_x}
     p_{gt}({x}) = \int  p_{X|Z}({x}|z)   p_Z(z) dz .
\end{equation}
\end{assumption}

\begin{assumption}\label{independence}
Assume $\tilde{Z}$ is independent of $X$ and $Z$, and the ground-truth missingness mechanism can be decomposed as
\begin{equation}\label{equ:gt_rx}
     p_{gt}(r|x) = \int    p_{R|X, \tilde Z}({r}|{x}, \tilde{{z}}) p_{\tilde Z}(\tilde{{z}}) d\tilde{{z}}.
\end{equation}
\end{assumption}
The above two assumptions lead to the decomposition of the full-data distribution:
\begin{equation}\label{equ:gt}
    p_{gt}({x},{r}) = \iint p_{X|Z}({x}|z)   p_{R|X,\tilde{Z}}({r}|{x}, \tilde{{z}}) p_Z(z) p_{\tilde{Z}}(\tilde{{z}}) dz d\tilde{{z}}.
\end{equation}
At this point, we present a theoretical justification for adopting this decomposition approach for modeling the full-data distribution in general missing data problems.
\begin{theorem}\label{universal}
Suppose that $\mathcal{X}$ is a Riemannian manifold diffeomorphic to $\mathbb{R}^d$, and that $\mu_{gt}$ is a probability measure on $\mathcal{X}$ which has a nowhere vanishing density function $p_{gt}(x)$ with respect to the volume measure.
\begin{enumerate}[(a)]
    \item If $d<p$, then for any $\epsilon > 0$, $r\in \{0,1\}^p$, and a measurable set $A \subset \mathbb{R}^p$ satisfying $\mu_{gt}(\partial A \cap \mathcal{X}) = 0$, where $\partial A$ is the boundary of $A$, there exist latent variables $Z \sim \mathcal{N}(0, I_{q_1})$ with $q_1 \geq d, \tilde{Z} \sim \mathcal{N}(0, I_{q_2})$ with $q_2 \geq 1$, a multivariate normal density $p_{X|Z}(x|z)$, and a multivariate Bernoulli distribution $p_{R|X, \tilde Z}(r|x, \tilde z)$ such that  
    \begin{equation}
         \left|\int_{x\in A}\iint p_{X|Z}({x}|z)   p_{R|X, \tilde Z}({r}|{x}, \tilde{{z}}) p_Z(z) p_{\tilde Z}(\tilde{{z}}) dz d\tilde{{z}}dx- \int_{\mathcal{X} \cap A} p_{gt}(r|x) \mu_{gt}(dx)\right| < \epsilon.
     \end{equation} 
    \item If $d=p$, then for any $\epsilon > 0, r \in \{0,1\}^p,$ and $x \in \mathcal{X}$, there are latent variables and distributions as above satisfying 
    \begin{equation}
         \left|\iint p_{X|Z}({x}|z)   p_{R|X, \tilde Z}({r}|{x}, \tilde{{z}}) p_Z(z) p_{\tilde Z}(\tilde{{z}}) dz d\tilde{{z}}-p_{gt}(r|x) p_{gt}(x) \right| < \epsilon.
     \end{equation} 
\end{enumerate}
\end{theorem}

Technical terms notwithstanding, Theorem \ref{universal} shows that in the missing data problem, the true full-data distribution can always be approximated arbitrarily well by an analytically tractable distribution in the form of $\int \int p_{X|Z}({x}|z)   p_{R|X, \tilde Z}({r}|{x}, \tilde{{z}}) p_Z(z) p_{\tilde Z}(\tilde{{z}}) dz d\tilde{{z}}$ . This implies that modeling the full-data distribution of $(X,R)$ using the decomposition in (\ref{equ:gt}) either accurately represents the true data generation process or approximates the distribution with high fidelity, thereby justifying its application. 

{\color{black}
\begin{remark}
We emphasize that Theorem 1 should be regarded as a motivating result for our latent variable formulation rather than as a central theoretical contribution. It primarily serves to illustrate the representational capacity of latent variable models, not to provide a fully general or exhaustive theoretical statement. While the theorem is presented under a diffeomorphism assumption for analytical simplicity, recent results such as \citet{dahal2022deep} and \citet{wang2025deep} demonstrate that mappings from simple base distributions to arbitrary target distributions can exist without this assumption, suggesting possible extensions beyond the present theoretical scope. Nevertheless, these generalizations fall beyond the present scope but may offer promising directions for future theoretical work.
\end{remark}
}

Henceforth, we will refer to (\ref{equ:gt}) as the ground-truth model on which we build our method throughout the paper.

\section{{Conditional} no self-censoring {given} latent variables
}\label{gt_model}


Without any assumption on the missingness mechanism, the identification of the joint distribution of $(X,R)$ is unachievable when data are MNAR.
To solve this identifiability issue, we make the following assumption about the missingness mechanism.
\begin{assumption}[{Conditional} no self-censoring {given} latent variables]\label{nscass}
    We assume that
    \begin{equation}\label{nsceq}
        R_j \ci{(X_j, R_{- j})}~ \Big| ~ X_{- j}, \tilde{Z} \quad \text{for} \quad j = 1, ..., p,
    \end{equation}
    where $A_{-j} = (A_1, ..., A_{j-1}, A_{j+1}, ..., A_p)$ denotes a subvector of $A$ with the $j^\text{th}$ entry removed. 

\end{assumption}


%


Assumption \ref{nscass} incorporates two parts: conditional independence among missingness indicators in $R$ and conditional independence between each missingness indicator $R_j$ and its corresponding data variable $X_j$.
In fact, conditional independence assumption among missingness indicators given $X$ and $\tilde{Z}$ is natural under the latent variable model for $R$ and the independence between $X$ and $\tilde{Z}$, as $X$ and $\tilde{Z}$ are supposed to capture all the dependence among missingness indicators.
Besides that, our Assumption \ref{nscass} requires that the conditional independence related to $R_j$ still holds even if we exclude $X_j$ from the given variables. This and the second part (conditional independence between each missingness indicator and its corresponding data variable) are where ``no self-censoring'' involved.

{
\color{black}
From a practical standpoint, there are many real-world situations in which a no self-censoring mechanism is plausible. For instance, in the clinical evaluation of Alzheimer’s disease, the decision to perform an invasive lumbar puncture depends on prior examination results, physician judgment, and patient preference, rather than on the unobserved values of the cerebrospinal fluid (CSF) biomarkers themselves \citep{lo2012predicting}. More generally, choices about whether to conduct additional diagnostic tests are typically informed by available clinical evidence, not by the unseen outcomes of the tests \citep{djulbegovic2015rational}. Such scenarios illustrate a no self-censoring mechanism, where missingness is independent of the variable itself once other relevant variables are taken into account.
}

\begin{figure}
     \centering
     \begin{subfigure}[a]{0.2\textwidth}
         \centering
         \vspace{-1.1in}
         \includegraphics[width=\textwidth]{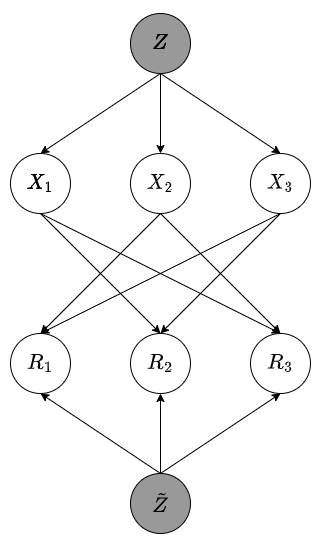}
     \end{subfigure}
    %
    \hspace{1in}
     \begin{subfigure}[b]{0.2\textwidth}
         \centering
         \includegraphics[width=\textwidth]{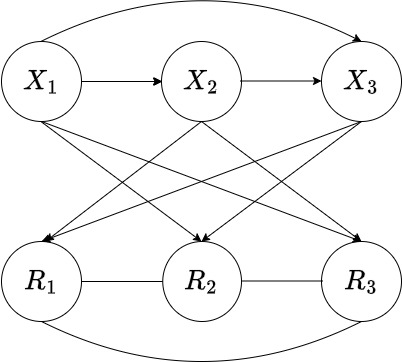}
     \end{subfigure}
     \vspace{-0.1in}
     \caption{Left: An example architecture of conditional no self-censoring given latent variables for 3 variables. Right: An example architecture of the traditional no self-censoring independence model for 3 variables. }\label{fig:nsc}
     \vspace{-10pt}
\end{figure}


An example architecture of our {conditional} no self-censoring model {given} latent variables is illustrated on the left panel of Figure \ref{fig:nsc} with $p=3$, where each missingness indicator $R_j$  is generated from $X_{-j}$ along with $\tilde{Z}$ for $j=1,2,3$. Note that given $X_{-j}$ and $\tilde{Z}$, Assumption \ref{nscass} excludes the conditional dependency between $X_j$ and its own missingness $R_j$, but does not exclude the marginal dependencies between them.

Compared to existing no self-censoring assumptions $R_j \ci{X_j}|X_{-{j}},R_{-{j}}$ 
and $R_j \ci{(X_j, R_{- j})}|X_{-{j}}$
\citep{sadinle2017itemwise, shpitser2016consistent, malinsky2021semiparametric,
mohan2013graphical, mohan2021graphical}, Assumption \ref{nscass} accounts for latent variables and requirements from latent variable models.
For a better comparison between these assumptions,
we provide an example architecture of the $R_j \ci{X_j}|X_{-{j}},R_{-{j}}$ assumption on the right panel of Figure \ref{fig:nsc}. An example of the $R_j \ci{(X_j, R_{- j})}|X_{-{j}}$ assumption can be achieved if we
eliminate the dependence of missingness indicators $R$ on the latent variables $\tilde{Z}$ in the left panel of Figure \ref{fig:nsc}.

Under our proposed conditional no self-censoring assumption, we establish the identifiability of the ground-truth missingness mechanism in the following theorem.

\begin{theorem}[Identifiability of missingness mechanism]\label{mechanism}
  The ground-truth missingness mechanism $p_{gt}(r|x)$ is nonparametrically identifiable under Assumptions \ref{independence} and \ref{nscass}.
\end{theorem}

Nonparametric identifiability essentially requires that there is a one-to-one mapping between the full-data distribution and the observed data distribution, so that the distribution of interest is a function of observed data only and not a function of unobserved data at all \citep{sadinle2018sequential, malinsky2021semiparametric}.
For the identification of the full-data distribution, an additional mild assumption on the probability of observing all data variables is required. 

\begin{assumption}[Positivity]\label{psass}
Suppose that
    \begin{equation}\label{pos}
        p_{gt}(R=\bm{1}|X) > c
    \end{equation}
    {with probability 1} for some constant $c > 0.$
\end{assumption}

\begin{corollary}\label{cor1}
    Under Assumptions \ref{independence}--\ref{psass}, the full-data distribution with ground truth $p_{gt}(x,r)$ is nonparametrically identifiable.
\end{corollary}

This result further implies that the ground-truth data distribution $p_{gt}(x)$ of $X$ is also nonparametrically identifiable. 
The identifiability not only serves as a desirable property itself, but also lays the foundation for the subsequent theoretical results of our proposed method on recovery of the full-data distribution.

}

{
\section{Method}\label{sec: method}
\subsection{The working model}
In this subsection, we propose to approximate the true full-data distribution $p_{gt}({x},{r})$ using a parametric working model.
The proposed working model is
\begin{equation*}
    p_{\theta, \psi}({x},{r}) = p_{\theta}({x})p_{\psi}({r}|{x}),
\end{equation*}
where 
\begin{equation}\label{x_fact}
    p_{\theta}(x) = \int \prod^p_{j=1} p_{\theta}(x_j|z) p(z)dz
    = \int p_{\theta}(x_{(r)}|z) p_{\theta}(x_{(m)}|z) p(z)dz
\end{equation}
and
\begin{equation}\label{equ:joint}
     p_{\psi}({r}|{x}) = \int \prod^p_{j=1} p_{\psi}({r_j}|{x_{-j}}, \tilde{{z}}) p(\tilde{{z}})  d\tilde{{z}}.
\end{equation}
Here 
$x_{(r)}$ and $x_{(m)}$ represent 
vectors consisting of observed and missing elements in $x$, respectively,
$p_{\theta}$ and $p_{\psi}$ are known distributions up to unknown parameters $\theta\in \Omega_{\theta}$ and $\psi\in \Omega_{\psi}$, respectively, and $p(z)$ and $p(\tilde{z})$ are density functions of latent variables $Z$ and $\tilde Z$, respectively, where $\Omega_{\theta}$ and $\Omega_{\psi}$ are parameter spaces.
We will provide specific forms for the known distributions $p_{\theta}$ and $p_{\psi}$ 
in {Section \ref{implementation}.}
In this working model, we let latent variables $Z\sim \mathcal{N}(0, I_{\kappa_1})$ and $\tilde{Z}\sim \mathcal{N}(0, I_{\kappa_2})$, where
the dimensions $\kappa_1$ and $\kappa_2$ of the $Z$ and $\tilde Z$ are not required to be the same as those ($q_1$ and $q_2$) in the ground-truth model since that information is lacking in practice. 

Note that  conditional independence  among $X$ given latent variable $Z$ is used in Equation (\ref{x_fact}), which enables us to separate observed and missing elements in $x$ given $Z$.
In fact, it does not exclude the correlations among $X$ themselves. 
The conditional no self-censoring in Assumption \ref{nscass} is preserved in Equation (\ref{equ:joint}) of the working model.
We will demonstrate in Section \ref{theory_recover} that the factorization in Equations (\ref{x_fact}) and (\ref{equ:joint}) does not prevent the proposed model from recovering the ground-truth full data distribution.

We propose to estimate the parameters $(\theta, \psi)$ through maximizing the log-likelihood of observed data under the working model:
\begin{equation}\label{log_p_obs}
    \log p_{\theta, \psi}(x_{(r)},r) = \log \iiint  p_{\theta}(x_{(r)}|z) p_{\theta}(x_{(m)}|z)\prod^p_{j=1} p_{\psi}(r_j|x_{- j}, \tilde{z})p(z)p(\tilde{z}) dz d\tilde{z} dx_{(m)}.
\end{equation}
However, the integral over missing and latent variables makes the maximization intractable. We will provide an efficient approximation to the log-likelihood in the following subsection.

\subsection{Variational inference}\label{sec:var_inf}
In this subsection, we introduce a variational inference method to solve the intractability issue of the integral in Equation (\ref{log_p_obs}).
Specifically, we propose variational distributions $q_{\phi}(z|x_{(r)}, r)$ and $q_{\lambda}(\tilde{z}|x_{(r)}, r)$ as approximations to the posteriors $p_{\theta, \psi}(z|x_{(r)}, r)$ and $p_{\theta, \psi}(\tilde{z}|x_{(r)}, r)$ in the working model, respectively,
where $\phi \in \Omega_{\phi}$ and $\lambda \in \Omega_{\lambda}$ are unknown parameters with parameter spaces $\Omega_{\phi}$ and $\Omega_{\lambda}$. We incorporate the variational distributions through rewriting the log-likelihood $\log p_{\theta, \psi}(x_{(r)},r)$ of the observed data  as
\begin{equation*}
   \log p_{\theta, \psi}(x_{(r)},r)= \log \mathbb{E}_{q_{\phi}(z|x_{(r)}, r) q_{\lambda}(\tilde{z}|x_{(r)}, r)p_{\theta}(x_{(m)}|z)} \left[ \frac{p_{\theta}(x_{(r)}|z) p_{\psi}(r|x_{(r)}, x_{(m)}, \tilde{z})p(z)p(\tilde{z})}{q_{\phi}(z|x_{(r)}, r)q_{\lambda}(\tilde{z}|x_{(r)}, r)}\right].
\end{equation*}
The evidence lower bound (ELBO)  \citep{kingma2013auto} for this log-likelihood is
\begin{equation*}
     \mathcal{L}_1(\theta, \psi, \phi, \lambda;x_{(r)}, r) = \mathbb{E}_{q_{\phi}(z|x_{(r)}, r) q_{\lambda}(\tilde{z}|x_{(r)}, r)p_{\theta}(x_{(m)}|z)} \left[\log \frac{p_{\theta}(x_{(r)}|z) p_{\psi}(r|x_{(r)}, x_{(m)}, \tilde{z})p(z)p(\tilde{z})}{q_{\phi}(z|x_{(r)}, r)q_{\lambda}(\tilde{z}|x_{(r)}, r)}\right].
\end{equation*}
For simplicity, we denote the fraction inside the logarithm as 
\begin{equation}\label{omega}
    w = \frac{p_{\theta}(x_{(r)}|z) p_{\psi}(r|x_{(r)}, x_{(m)}, \tilde{z})p(z)p(\tilde{z})}{q_{\phi}(z|x_{(r)}, r)q_{\lambda}(\tilde{z}|x_{(r)}, r)}.
\end{equation}

Similar to the approach proposed in \citet{ipsen2020not} inspired by the importance-weighted variational inference in \cite{burda2015importance},
we improve the above ELBO through
replacing the fraction inside the logarithm in the ELBO with a Monte Carlo estimate of it. Thus our proposed objective function is
\begin{equation}\label{objective}
    \mathcal{L}_K (\theta, \psi, \phi, \lambda;x_{(r)}, r) = \mathbb{E}_{\prod_{k=1}^K q_{\phi}(z_k|x_{(r)}, r) q_{\lambda}(\tilde{z}_k|x_{(r)}, r)p_{\theta}(x_{(m),k}|z_k)} \left[\log \frac{1}{K} \sum^K_{k=1} w_k\right],
\end{equation}
where 
\begin{equation*}
    w_k =\frac{p_{\theta}(x_{(r)}|z_k) p_{\psi}(r|x_{(r)}, x_{(m),k}, \tilde{z}_k)p(z_k)p(\tilde{z}_k)}{q_{\phi}(z_k|x_{(r)}, r)q_{\lambda}(\tilde{z}_k|x_{(r)}, r)},
\end{equation*}
$X_{(m),k}$, $Z_k$ and $\tilde{Z}_k$ follow the same distributions as $X_{(m)}$, $Z$ and $\tilde{Z}$, respectively,
for $1 \leq k \leq K.$ Note that when $K = 1$, the lower bound (\ref{objective}) coincides with ELBO.
In Section \ref{theory_bounds}, we will show that 
$\mathcal{L}_K$ is a tighter lower bound for the log-likelihood than $\mathcal{L}_1$ for $K > 1$.
 The lower bounds of the log-likelihood based on variational distributions approximate the log-likelihood {in the working model} while circumventing the issue of intractability. \textcolor{black}{Despite our common foundation with the model of not-MIWAE by \citet{ipsen2020not} built on importance weighted variational inference and missingness mechanisms, we note a main difference of the proposed objective function:  we have an additional latent variable $\tilde{Z}$ for $R$, and hence an additional $q_{\lambda}(\tilde{z}|x_{(r)}, r)$ as approximations to the posterior $p_{\theta, \psi}(\tilde{z}|x_{(r)}, r)$, and prior $p(\tilde{z})$, along with the conditional independence of $R_j$ and $X_j$ given $\tilde{Z}$. These are critical for incorporating correlations among the missingness indicators and identifiability. }
 
Since the expectation in $\mathcal{L}_K$ is unachievable from samples, we define the per-sample estimator of $\mathcal{L}_K$,
\begin{align*}
    \hat{\mathcal{L}}_K=\log \frac{1}{K} \sum^K_{k=1} w_k,
\end{align*}
and use $\hat{\mathcal{L}}_K$ as our objective in practice.
 We estimate the variational parameters ($\phi, \lambda$) and the generative parameters $(\theta, \psi)$ through maximizing $\hat{\mathcal{L}}_K$. 

{\color{black}We note that this working parametric deep latent variable model architecture explicitly incorporates the proposed no self-censoring condition (Assumption \ref{nscass}) through Equation (\ref{equ:joint}), which directly constrains the log-likelihood, its lower bound, and the empirical lower bound that defines the objective function. Consequently, the condition required for identifiability is enforced during learning, 
 guiding the model to generalize to the shape of the complete-data distribution.}

\subsection{Implementation}\label{implementation}
In this subsection, we provide implementation details of the proposed method based on the stochastic gradient descent algorithm.
As in the architectures of VAEs, we refer to the inference models $q_{\phi}(z|x_{(r)}, r)$ and $q_{\lambda}(\tilde{z}|x_{(r)}, r)$ as {\it probabilistic encoders}, and $p_{\theta}(x|z)$ and $p_{\psi}(r|x, \tilde{z})$ as {\it probabilistic decoders},
for which we have the following conventions. 



    The encoders are set to take the form of multivariate Gaussian such that $q_{\phi}(z|x_{(r)}, r) = \mathcal{N}(\mu_z, \Sigma_z)$ and $q_{\lambda}(\tilde{z}|x_{(r)}, r) = \mathcal{N}(\mu_{\tilde{z}}, \Sigma_{\tilde{z}}),$ where $\mu_z = f_{\mu_z}(x_{(r)}, r;\phi)$ and $\Sigma_z = S_zS_z^T$ with $S_z = f_{S_z}(x_{(r)}, r;\phi)$, and $\mu_{\tilde{z}} = f_{\mu_{\tilde{z}}}(x_{(r)}, r;\lambda)$ and $\Sigma_{\tilde{z}} = S_{\tilde{z}}S_{\tilde{z}}^T$ with $S_{\tilde{z}} = f_{S_{\tilde{z}}}(x_{(r)}, r;\lambda)$.
For the decoders, a Gaussian decoder is adopted for the data: $p_{\theta}(x|z) = \mathcal{N}(\mu_x, \Sigma_x),$ with decoder moments $\mu_x = f_{\mu_x}(z;\theta)$ and $\Sigma_x = \gamma I$ for a tunable scalar $\gamma >0.$  For the binary missingness indicators, a Bernoulli decoder is used for each dimension $1 \leq j \leq p$: $p_{\psi}(r_j|x, \tilde{z}) = \text{Bernoulli}(\pi_j),$ 
where $\pi_j = f_{r_j}(x,\tilde{z};\psi)$. 
Here $f_{\mu_z}, f_{S_z}, f_{\mu_{\tilde{z}}}, f_{S_{\tilde{z}}}, f_{\mu_x} $ and $f_{r_j} $ specify parameterized functional forms that can be arbitrarily complex, and we use deep neural networks to approximate these functions.
%

For the missing part of the input to neural networks, we resort to an initial zero imputation following the common practice in \cite{nazabal2020handling, ipsen2020not,ma2021identifiable}. 
In the generative missingness model $p_\psi (r|x,\tilde{z})$, we use a mixture of the observed part of the data $x_{(r)}$ and the generated data for the original missing part $\hat{x}_{(m)}$ from the generative model $p_{\theta}(x|z)$. 
Parameters $\{\theta, \psi, \phi, \lambda\}$ are estimated by stochastic gradient descent, with reparameterization tricks \citep{kingma2013auto} for latent random variables and missing data so that Monte Carlo estimates of expectations with respect to $q_{\phi}(z|x_{(r), r})  q_{\lambda}(\tilde{z}|x_{(r)}, r) p_{\theta}(x_{(m)}|z)$ are differentiable. We summarize the algorithm for the implementation of the proposed method in Algorithm \ref{alg:cap}.



\begin{algorithm}[t]\small%
\caption{\textcolor{black}{The proposed IM-IWAE}}\label{alg:cap}
\begin{algorithmic}
\Require Missingness indicators $r^{(1)}, \cdots, r^{(n)}$, observed data $x^{(1)}_{(r^{(1)})}, \cdots, x^{(n)}_{(r^{(n)})}$,  batch size $b$, learning rate $\alpha$, number of importance samples $K$, decoder variance $\gamma$, maximum number of iterations $l_{\max}$,
deterministic neural networks: encoders $f_{\mu_z}, f_{S_z}, f_{\mu_{\tilde{z}}}, f_{S_{\tilde{z}}}$, decoders $f_{\mu_x}$, $f_{r_j}$ for $1 \leq j \leq p$
\Ensure $\phi, \lambda, \theta, \psi$
\State $\phi^{[0]}, \lambda^{[0]}, \theta^{[0]}, \psi^{[0]} \gets $random initialization
\For{$l \gets 1$ to $l_{\max}$ }
\State Select a batch $\{(x^{(1)}_{(r^{(1)})},r^{(1)}), \cdots, (x^{(b)}_{(r^{(b)})},r^{(b)})\}$ from the data at random
{
\State Let $\tilde{x}^{(i)}_{(m^{(i)})} = 0$ for $1 \leq i \leq b$, and $\tilde{X} = \{\tilde{x}^{(1)}, \cdots, \tilde{x}^{(b)}\}$
}
\For{each $\tilde{x}^{(i)} \in \tilde{X}$}
\State Generate noise $\epsilon_1, \cdots, \epsilon_K \overset{\text{i.i.d}}{\sim} \mathcal{N}(0,I_{\kappa_1 + \kappa_2})$%
\State $(\mu_z^T,\mu_{\tilde{z}}^T)^T \gets$  $\left(f^T_{\mu_z}(\tilde{x}^{(i)}; \phi^{[l-1]}), f^T_{\mu_{\tilde{z}}}(\tilde{x}^{(i)}; \lambda^{[l-1]})\right)^T$, with $\mu_z \in \mathbb{R}^{\kappa_1}, \mu_{\tilde{z}} \in \mathbb{R}^{\kappa_2}$, 
\State $(\sigma_z^T, \sigma_{\tilde{z}}^T)^T \gets$ $\left(f^T_{S_z}(\tilde{x}^{(i)}; \phi^{[l-1]}), f^T_{S_{\tilde{z}}}(\tilde{x}^{(i)}; \lambda^{[l-1]})\right)^T$, with $\sigma_z \in \mathbb{R}^{\kappa_1}, \sigma_{\tilde{z}} \in \mathbb{R}^{\kappa_2}$
\For{$k \gets 1$ to $K$ }
\State $(z_k^T, \tilde{z}^T_k)^T \gets  (\mu_z^T, \mu_{\tilde{z}}^T)^T + (\sigma_z^T, \sigma_{\tilde{z}}^T)^T \odot \epsilon_k$, where $\odot$ denotes Hadamard product
\State $\mu_x \gets$ $f_{\mu_x}\left(z_k; \theta^{[l-1]}\right)$,  $\hat{x}^{(i)}_k \gets \mathcal{N}\left(\mu_x,\gamma I\right)$
\For{$j \gets 1$ to $p$}
\State $\hat{r}^{(i)}_{j,k} \gets f_{r_j}\left((\bar{x}^{(i)}_k)_{-j}, \tilde{z}_k;\psi^{[l-1]}\right)$, where $ \bar{x}^{(i)}_k = \tilde{x}^{(i)} +\hat{x}^{(i)}_k \odot (1-r^{(i)})$
\EndFor
\EndFor
\State 
$\mathcal{L}_K^{(i)} = \log \frac{1}{K}\sum^K_{k=1}w_k^{(i)}$, where 
$w_k^{(i)} =\frac{p_{\theta^{[l-1]}}(\hat{x}^{(i)}_{(r^{(i)})}|z_k) p_{\psi^{[l-1]}}(\hat{r}^{(i)}_k|\bar{x}^{(i)}_k, \tilde{z}_k)p(z_k)p(\tilde{z}_k)}{q_{\phi^{[l-1]}}(z_k|x^{(i)}_{(r^{(i)})}, r^{(i)})q_{\lambda^{[l-1]}}(\tilde{z}_k|x^{(i)}_{(r^{(i)})}, r^{(i)})}$

\EndFor

\State $(\phi^{[l]}, \lambda^{[l]}, \theta^{[l]}, \psi^{[l]})^T \gets (\phi^{[l-1]}, \lambda^{[l-1]}, \theta^{[l-1]}, \psi^{[l-1]})^T - \alpha \cdot
{\nabla \sum^b_{i=1}\mathcal{L}^{(i)}_K(\phi^{[l-1]}, \lambda^{[l-1]}, \theta^{[l-1]}, \psi^{[l-1]})}$

\If{convergence}
\State break
\EndIf
\EndFor
\end{algorithmic}
\end{algorithm}

\subsection{Data imputation}\label{data_imp}
We propose to impute  missing values through 
$\underset{x_{impute}}{\min}\mathbb{E}_{x_{(m)}}[L(x_{(m)}, x_{impute})|x_{(r)},r]$, where $L$ is a criterion function. 
When $L$ corresponds to the squared error, the {proposed} imputation will be the conditional mean 
\begin{equation}\label{impute_formula}
    \mathbb{E}[x_{(m)}|x_{(r)},r] = \iiint x_{(m)} \frac{p_{\psi}(r|x_{(r)}, x_{(m)},\tilde{z}) p_{\theta}(x_{(r)}|z) p(z) p(\tilde{z})}{q_{\phi}(z|x_{(r)}, r)q_{\lambda}(\tilde{z}|x_{(r)}, r)} d z d \tilde{z} d x_{(m)}.
\end{equation}
After obtaining the estimated parameters $(\hat{\phi}, \hat{\lambda}, \hat{\theta}, \hat{\psi})$ through Algorithm \ref{alg:cap}, we use the self-normalized importance sampling \citep{mattei2019miwae,ipsen2020not}, with the proposal $q_{\hat{\phi}}(z|x_{(r)}, r)q_{\hat{\lambda}}(\tilde{z}|x_{(r)}, r)p_{\hat{\theta}}(x_{(m)}|z)$ to obtain an estimate
\begin{equation*}
    \hat{x}_{(m)} = \mathbb{E}[x_{(m)}|x_{(r)},r] \approx \sum^B_{b = 1} \alpha_b x_{(m), b}, \text{ with } \alpha_b = \frac{w_b}{w_1 + ... + w_B},
\end{equation*}
{where for $1 \leq b \leq B,$
\begin{equation*}
    w_b = \frac{p_{\hat{\theta}}(x_{(r)}|z_b) p_{\hat{\psi}}(r|x_{(r)}, x_{(m),b}, \tilde{z}_b)p(z_b)p(\tilde{z}_b)}{q_{\hat{\phi}}(z_b|x_{(r)}, r)q_{\hat{\lambda}}(\tilde{z}_b|x_{(r)}, r)},
\end{equation*}
and
 $\{(z_b, \tilde{z}_b, x_{(m),b})\}_{b=1}^B$ are $B$ $i.i.d.$ samples from $q_{\hat{\phi}}(z|x_{(r)}, r)q_{\hat{\lambda}}(\tilde{z}|x_{(r)}, r)p_{\hat{\theta}}(x_{(m)}|z)$.} 

In addition to missing data imputation, the proposed model can generate synthetic data from the latent variables through the learned distribution of the data.
This is an important advantage of generative models in the presence of missing data, since the generated data can be used to gain more insight into the distribution itself in addition to predicting missing values using the observed data, which we will illustrate in Section \ref{simulationstudy}.

\subsection{On the choice of latent dimensions}\label{cld}
In real-world data scenarios, where the true latent dimensions $q_1$, $q_2$ are unknown, the latent dimensions in the learning method, $\kappa_1$, $\kappa_2$, must be empirically set to chosen values. 
{
We have observed that varying the latent dimension for the missingness mechanism, $\kappa_2$, does not significantly impact performance, with $\kappa_2 = 1$ generally performing as well as higher values (See Appendix \ref{app_dimzt} for more details). \textcolor{black}{This aligns with our theoretical results and the proof in \ref{appx_mech}: 
a single continuous latent dimension for the missingness mechanism suffices to approximate the joint distribution of binary variables. Additional dimensions add no benefit, as they can be ignored in recovering the joint distribution. }
\textcolor{black}{In contrast, for $\kappa_1$, the dimension of the latent variable corresponding to the continuous data, setting its value just above the intrinsic dimensionality generally improves performance compared with smaller values, as it allows the latent space to represent the data distribution more faithfully under practical network capacity \citep{wang2025deep}.}

For 
$\kappa_1$, we propose the following  cross-validation selection procedure. One challenge in applying cross-validation to real-world data with missing values is that not all data points are available. To address this, we propose the following approach: for each training and validation dataset, both containing missing values, we train the model on the training data and mask some observed values in the validation data completely at random. We then use the trained model to impute the missing values in the validation set (including both the originally missing and synthetically masked values) and calculate the imputation errors solely on the synthetically masked entries.

Recognizing that these synthetically masked missing values are ignorable, we impute them using a different estimate, $\mathbb{E}[x_{(m)}|x_{(r)}]$, which does not depend on the values of the missingness indicators as in (\ref{impute_formula}). The value of $\kappa_1$ that minimizes the imputation errors on average across the validation sets is selected. See Figure \ref{fig:cv} in Section \ref{realexp} for an example of choosing the latent dimensions in the real data.


\section{Theory}\label{theory}
In this section, we establish asymptotic properties of the proposed objective function $\mathcal{L}_K$, as well as the capacity of the proposed method to recover ground-truth distributions.
All proofs are postponed to the supplementary material.

\subsection{Convergence of lower bounds}\label{theory_bounds}
In this subsection, we investigate theoretical limits on the approximation of the objective function to the log-likelihood of the observed data under the working model, along with the bias and variance of the objective function when we consider it as an estimator of the log-likelihood. 
\begin{theorem}\label{LK}
    For all $K \geq 1$, the lower bounds satisfy 
    \begin{equation*}
        \mathcal{L}_{K+1} (\theta, \psi, \phi, \lambda;x_{(r)}, r)\geq  \mathcal{L}_{K} (\theta, \psi, \phi, \lambda;x_{(r)}, r),
    \end{equation*}
    and
    \begin{equation*}
        \log p_{\theta, \psi}(x_{(r)},r) \geq  \mathcal{L}_{K} (\theta, \psi, \phi, \lambda;x_{(r)}, r) \quad \forall \phi \in \Omega_\phi, \lambda\in \Omega_\lambda.
    \end{equation*}
    Moreover, if $w$ in Equation (\ref{omega}) is bounded, 
    we have
    \begin{equation*}
        \log p_{\theta, \psi}(x_{(r)},r) = \lim_{K\rightarrow \infty} \mathcal{L}_{K} (\theta, \psi, \phi, \lambda;x_{(r)}, r) \quad \forall \phi \in \Omega_\phi, \lambda\in \Omega_\lambda.
    \end{equation*}
    
\end{theorem}



Theorem \ref{LK} demonstrates the convergence of the {lower} bound (the population-version objective function $\mathcal{L}_K$) of the proposed IM-IWAE method to the log-likelihood. 
A distinction of our bound from the one for not-MIWAE is that we integrate over two latent spaces for not only the data but also the missingness mechanism {\citep{ipsen2020not}}. Such a difference does not fundamentally alter the property of the bound in terms of the approximation to the log-likelihood.


In the following, we provide the bias and variance of the empirical-version objective function $\hat{\mathcal{L}}_K$ when we treat it as an estimate of the log-likelihood $\log p_{\theta, \psi}(x_{(r)},r)$.


\begin{proposition}[Bias and variance of $\hat{\mathcal{L}}_K$]\label{prop1}
Suppose that the distribution of $w$ defined in Equation (\ref{omega}) is
    supported on the positive real line with finite moments of every order.
    As $K \rightarrow \infty$, we have the following asymptotic results for $\hat{\mathcal{L}}_K$ as an estimator of $\log p_{\theta, \psi}(x_{(r)},r)$:
    \begin{equation*}
        \text{Bias}(\hat{\mathcal{L}}_K) = - \frac{1}{K}\frac{\mu_2}{2 \mu^2} + \frac{1}{K^2}\left( \frac{\mu_3}{3\mu^3} - \frac{3\mu_2^2}{4\mu^4}\right)  + o(K^{-2}),
    \end{equation*}
    and
    \begin{equation*}
        \text{Var}(\hat{\mathcal{L}}_K) = \frac{1}{K}\frac{\mu_2}{\mu^2} - \frac{1}{K^2} \left( \frac{\mu_3}{\mu^3} - \frac{5\mu^2_2}{2\mu^4}\right) + o(K^{-2}),
    \end{equation*}
    where $\mu_i \coloneq \mathbb{E}[(w - \mathbb{E}[w])^i]$ 
    and $\mu \coloneq \mathbb{E}_P[w]$. 
\end{proposition}

For each $(\theta, \psi)$, Proposition \ref{prop1} shows that both the bias and the variance of the estimator $\hat{\mathcal{L}}_K$ decrease at a rate of $O(1/K)$. Note that the leading  terms involves $\mu_2/\mu^2$, which is the variance of the distribution of $w$  divided by the squared mean, also known as the squared \emph{coefficient of variation}. Hence, for large $K$, the bias and the variance of $\hat{\mathcal{L}}_K$ are small when the coefficient of variation is low. 
This leads to the following convergence of $\hat{\mathcal{L}}_K$.
\begin{corollary}[Convergence of $\hat{\mathcal{L}}_K$]\label{consistency}
   For all $\epsilon >0$, 
   \begin{equation*}
       \lim_{K \to \infty} P(|\hat{\mathcal{L}}_K -\log p_{\theta, \psi}(x_{(r)},r) | \geq \epsilon) = 0
   \end{equation*}
\end{corollary}

\subsection{Recovery of the ground-truth full-data distribution}\label{theory_recover}
In this subsection, we establish the theoretical framework at the population level for recovering the  ground-truth full-data distribution using the proposed working model,
towards which a crucial step is to demonstrate that the proposed model can accurately approximate the ground-truth distribution of missingness indicators.


{\color{black}
\begin{lemma}\label{lemma_r}
	When we consider a Bernoulli distribution with mean $f_{r_j}(x_{-j},\tilde{z};\psi)$ as the decoder $p_{\psi}(r_j|x_{-j}, \tilde{z})$ of the missingness mechanism for $j=1,\dots, p$, 
 there exist $f_{r_j}$’s with parameter $\psi^*$ that recover the ground truth such that 
	\begin{equation*}
		p_{\psi^*}(r|x) = p_{gt}(r|x) .
	\end{equation*}
\end{lemma}
}


Lemma \ref{lemma_r} theoretically justifies the capacity of the proposed model to recover the ground-truth missingness mechanism by means of a factorization among the missingness indicators given the data $X$ and the latent variable $\tilde{Z}$. 
Based on the result for the capacity with respect to the decoder in Lemma \ref{lemma_r}, we can show that the underlying true joint distribution of $X_{(R)}$ and $R$ can be well approximated by our latent variable models, as stated in the following theorem.

{
\begin{theorem}\label{final_thm}
Consider the proposed IM-IWAE method with the objective function in Equation (\ref{objective}).
Suppose the ground-truth density $p_{gt}(x)$ is nonzero everywhere on $\mathbb{R}^p$. Assume that $w$ is bounded. Then for any chosen latent dimension $\kappa_1$ and $\kappa_2$ for the method with $\kappa_1 \geq p$ and $\kappa_2 \geq 1$, there is a sequence of decoders and encoders with parameters $\{\theta^*_{t}, \psi^*_{t}, \phi^*_{t}, \lambda^*_{t}\}$ such that 
    \begin{equation*}
{\lim_{t\to\infty}}\lim_{K\to\infty}\mathcal{L}_K (\theta^*_{t}, \psi^*_{t}, \phi^*_{t}, \lambda^*_{t}; x_{(r)}, r) = \log p_{gt}(x_{(r)}, r).
    \end{equation*}
\end{theorem}
}
{\color{black}We note that Theorem \ref{final_thm}, our central result on recovering data distributions with the working models, holds independently of the manifold assumption in Theorem \ref{universal}, as Theorem \ref{universal} only serves to motivate the use of latent variable models.}
Theorem \ref{final_thm} suggests that, given a sufficiently large $K$, we can approximate the ground-truth joint distribution of the observed data and the missingness indicators 
through the parameterized objective function,
regardless of conditions on the missingness mechanism. Theorem \ref{mechanism} (identification) and Theorem \ref{final_thm} (recovery of the true distribution $p_{gt}(x_{(r)}, r)$) together ensure the recovery of the true full-data distribution $p_{gt}(x,r)$ including the missing part, which is the target of our method, as the identification specifies a one-to-one mapping between the joint distribution $p_{gt}(x_{(r)}, r)$ of the observed part and the full-data distribution $p_{gt}(x,r)$.

}

\section{Simulations}\label{simulationstudy}
In this section, we provide a series of simulation studies to evaluate our approach. 
Our evaluation encompasses various aspects, including data imputation, data generation, and estimation for a parameter of interest. 
To benchmark the effectiveness of our proposed methods, we compare them with two contemporary deep generative models for MNAR data,  \textcolor{black}{GINA} \citep{ma2021identifiable} and \textcolor{black}{not-MIWAE} \citep{ipsen2020not}. Furthermore, we include two classical imputation methods, \textcolor{black}{MICE} \citep{van2011mice} and \textcolor{black}{missForest} \citep{stekhoven2012missforest}, which are based on the commonly used MAR assumptions, as our baseline references. We adopt the same neural network structures and hyperparameter settings for methods based on generative models.


\subsection{Imputation and generation}\label{sim3d}
We first focus on 3-dimensional missing data scenarios with no self-censoring, and then consider more complex cases later in this subsection.


\noindent{\bf Generation of complete data.} The data generation process begins with latent variables $Z_1, Z_2, Z_3 \overset{\text{i.i.d}}{\sim} \mathcal{N}(0,1).$ Subsequently, we generate each data variable as 
\begin{equation*}
    X_j = f_{\theta_j}(Z_1,Z_2,Z_3)+\epsilon_j, 
\end{equation*}
for $j = 1,2,3$, where $f_{\theta_j}$ represents \textcolor{black}{a nonlinear mapping with one hidden layer coupled with the $tanh$ function and randomly generated coefficients}, and $\theta_j $ denotes distinct parameter sets for $f_{\theta_j}$. And $\epsilon_1, \epsilon_2, \epsilon_3 \overset{\text{i.i.d}}{\sim} \mathcal{N}(0,0.01)$ serve as noise variables. 
%
For the missingness mechanism, we simulate missingness under two different strategies.

\noindent{\bf Simulation of missingness with latent variables.} For the first strategy, we adopt the scheme 
\begin{equation*}
    R_j \sim \text{Bernoulli} \big(h_{\psi_j}(X_{-j}, \tilde{Z})\big), ~~~~ 1 \le j \le 3,
\end{equation*}
 involving a latent variable $\tilde{Z} \sim \mathcal{N}(0,1) $ and a mapping function $h_{\psi_j}$ \textcolor{black}{for which we experiment with both linear and nonlinear structures, where the coefficients are randomly chosen, and the $tanh$ function is used for the nonlinear case}.  Note that in this scenario, all variables $X_1, X_2, X_3$ could be missing and that auxiliary variables are unavailable. 


\noindent{\bf Simulation of missingness without latent variables.}   We also examine a more general case, where we simulate missingness according to the no self-censoring condition without resorting to a latent variable. We establish a dependency of the missingness indicators on the other data variables while preserving the correlation between these indicators, as in the right panel of Figure \ref{fig:nsc}. This is achieved using an independent noise variable according to the noise-outsourcing lemma in probability theory \citep{kallenberg1997foundations, austin2015exchangeable}. Specifically, when generating missingness, for each dimension $j\in \{1,2,3\}$ of the missingness indicators, we sample from a uniform distribution $U \sim \mathcal{U}(0,1)$ and apply a transformation $g_{\psi_j}(X_{-j})$ with randomized parameter $\psi_j$, where 
$g_{\psi_j}$ is a linear function or a nonlinear function with \emph{tanh} activations.
We let $R_j=1$ if $U>g_{\psi_j}(X_{-j})$, and $R_j=0$ otherwise.


In each missingness strategy, we simulate a total of $20,000$ random samples in each replication, with a total of $50$ replications conducted. The simulated missing rates are controlled to fall within the range of $30\%$ to $40\%$ \textcolor{black}{by adding a constant to the logits of Bernoulli distribution for the missingness mechanism}. 
Samples in which all variables are missing are excluded from the training set.
%

\sisetup{detect-weight,mode=text}
\renewrobustcmd{\bfseries}{\fontseries{b}\selectfont}
\renewrobustcmd{\boldmath}{}
\newrobustcmd{\B}{\bfseries}
\def\code#1{\texttt{#1}}
\addtolength{\tabcolsep}{-4.1pt}

\begin{table}[t]

\begin{center}
\caption{Imputation RMSE of simulated missing data with no self-censoring}
\begin{tabular}{l|cc|cc}
\hline
\multirow{2}{*}{Method} & \multicolumn{2}{c|}{With latent variables}& \multicolumn{2}{c}{Without latent variables}\\
& linear & nonlinear& linear & nonlinear \\ \hline
Proposed &  \B 0.7081 (0.1146) &  \B 0.7145 (0.1325)&\B 0.7386 (0.1687) &  \B 0.7188 (0.1186) \\
GINA &  0.8682 (0.2309) & 0.8676  (0.1324)& 0.9957(0.3484) & 0.8650(0.1062)  \\
Not-MIWAE & 0.9311(0.6139)  & 0.7570(0.1619)& 0.8877(0.3558)  & 0.7423(0.1230) \\
MICE & 1.2544(0.3144) & 1.0498(0.2255) &1.1998(0.2794) & 1.0772 (0.2370) \\
MissForest & 0.9718(0.2198)& 0.9507(0.1400) & 0.9998(0.1694) & 0.9506(0.1208) \\
\hline
\end{tabular}%
\label{3dsim}
\end{center}
\vspace{-20pt}
\end{table}



\noindent{\bf Results.}
We evaluate the performance of each method under each setting via
the imputation root mean squared error (RMSE) which is provided in 
Table \ref{3dsim}.
In the first setting where latent variables are used for the missingness mechanism, there is no model misspecification for our method. 
As shown on the left side of Table \ref{3dsim}, the proposed method significantly outperforms GINA which also incorporates latent variables for the missingness indicators yet with a modeling of self-censoring. Not-MIWAE shows a less stable performance in terms of both mean and standard deviation. This instability is likely due to 
its inability to account for the correlation among missingness indicators and the exclusion of self-censoring. It is noteworthy that neither of GINA and not-MIWAE  guarantees model identifiability under this MNAR scenario, which may have contributed to their inferior performance in our experiments.  \textcolor{black}{The three deep generative models designed for MNAR data clearly outperform the classical methods (MICE and MissForest) based on MAR assumptions in terms of lower imputation RMSE.}


\textcolor{black}{In the setting where missingness is generated without  latent variables, as shown on the right side of Table \ref{3dsim}, our imputation errors increase slightly compared to the previous case probably due to the model misspecification of the missingness mechanism. However, our method still outperforms other methods overall. In contrast, GINA exhibits deteriorated imputation performance similar to the baseline methods for MAR data, and increased variance, in the linear case, which could be possibly due to its nonidentifiability issue that our model has successfully addressed.} These results demonstrate advantages of our proposed method over existing methods and certain robustness in terms of missing value imputation in different scenarios of no self-censoring.

\begin{figure}[t]
    \centering
    \includegraphics[scale = 0.6]{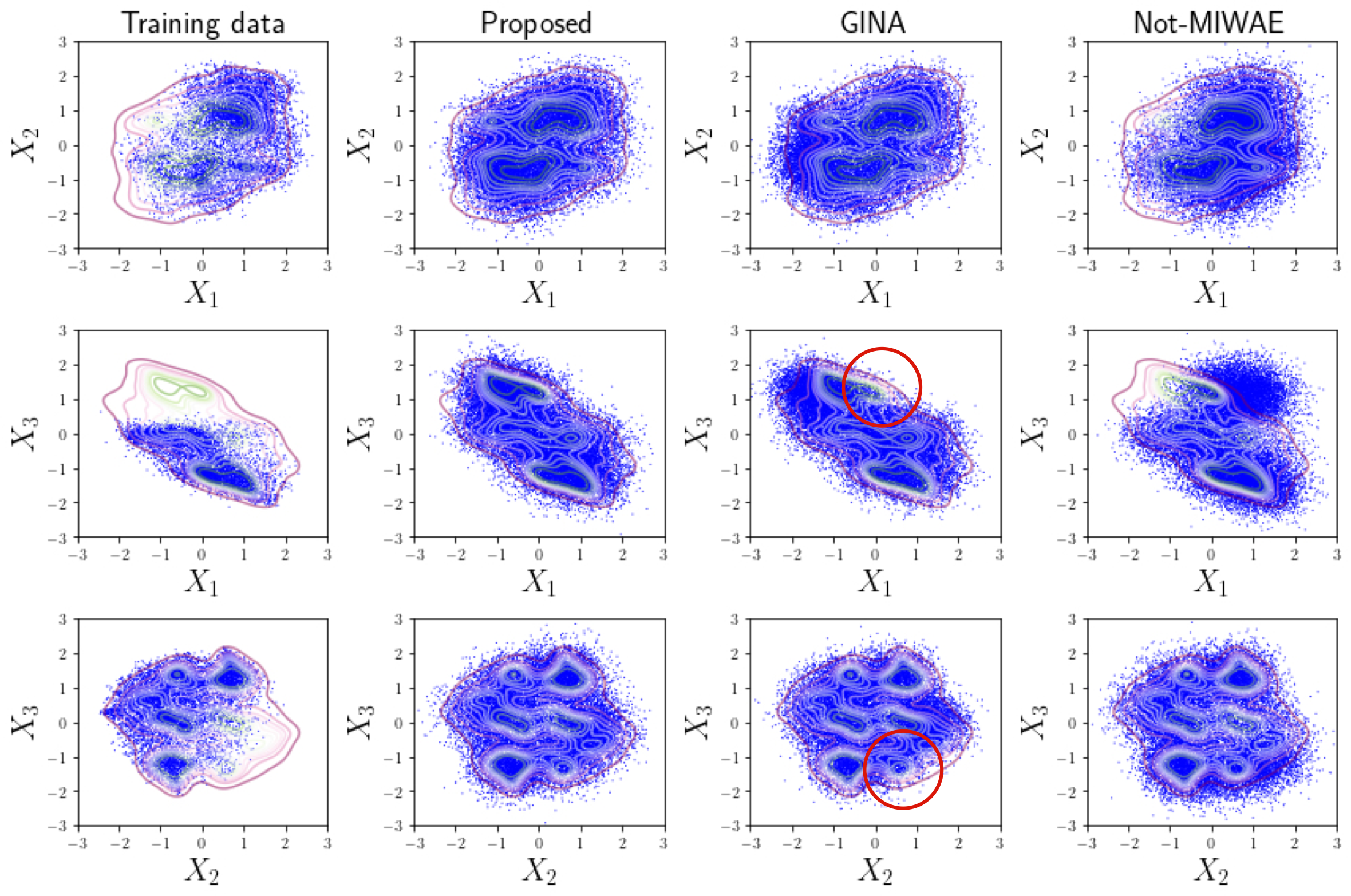}
    \caption{Visualization of generated data from the deep generative models under the setting with latent variables and a linear transformation $g_{\psi_j}$. 
}
    \label{fig:plot2d}
    \vspace{-20pt}
\end{figure}

\noindent {\bf Visualization.}
We plot the generated data distributions of the three deep generative models in Figure \ref{fig:plot2d}, and compare them with
the original complete data and the training data with missing values.
%
%
 We examine the data in all 2-dimensional combinations of the 3-dimensional data. 
 {The contour lines in Figure \ref{fig:plot2d} represent the density of the true data distribution, where pink represents lower values,
white ones represent middle values,
and green ones represent higher values. The blue dots in Figure \ref{fig:plot2d} represent samples in the training data or from the imputation based on deep generative models.} 

Due to MNAR, samples in the training data are a biased representation of the true data distribution. Although all three generative models have made a clear effort to recover the original distribution by including a missingness model, not-MIWAE produced a substantial bias. The performances of the proposed model and GINA do not appear to differ much, but a closer examination of the areas in the red circles in Figure \ref{fig:plot2d} reveals that GINA tends to miss some of the modes in the data distribution, while the proposed model captured those more closely. In addition, calculations of the maximum mean discrepancy (MMD) against the true data distribution over $50$ replications show that the proposed model has a mean $0.0086$ with a standard deviation $0.0074$, while GINA has a higher mean of $0.0098$ with a larger standard deviation $0.0102$, again demonstrating the superior learning and generative performance of our model in this scenario. 

\begin{figure}[tbh!]
    \centering
    \includegraphics[scale = 0.6]{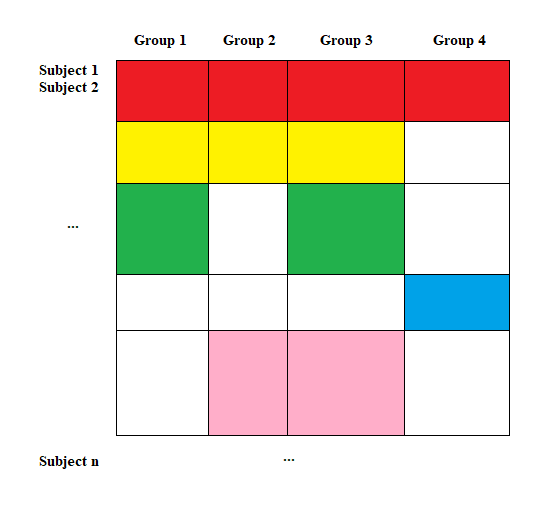}
    \caption{An illustration of block-wise missingness with $4$ groups for $n$ subjects. Each group is a set of variables either missing together or observed together. Each row represents a subject, that is, a sample of $X$. Here white areas represent blocks of missing values. 
    Different colors represent different missing patterns. Note that subjects are arranged in an order such that they have the same missing pattern. For the sake of brevity, we present only a subset of the missing patterns, omitting the others. (Note that the group sizes are not required to be the same, where the group size means the number of variables within a group.The block-wise missingness with $4$ groups correspond to the case with $p=100$ and $gs=25$, the last column in Table \ref{highdimsim} for the experiments.) 
    } 
    \label{fig:block}
    \vspace{-20pt}
\end{figure}

\noindent {\bf Higher dimensions}
The low-dimensional simulation enabled us to closely examine the performance of the method, both numerically and visually, in comparison with other methods across different scenarios. Applying similar settings to higher dimensions provides valuable insights into the model's performance in more challenging situations.

For higher dimensions, we generated $X$ from latent variables, experimenting with data dimensions of \( p = 30, 50, 100 \) while keeping the latent dimension fixed at 10. All other parameters and settings were consistent with the previous $3$-dimensional experiments, with no latent variable used for generating missingness and with the mapping function $h$ being linear. The missing rates were controlled to be around $50\%$. In these higher-dimensional simulations, we considered different settings with block-wise missingness, where variables within the same group are either missing together or observed together. See Figure \ref{fig:block} for an illustration of a dataset with $3$ groups. Specifically, we varied the sizes of the groups. The term ``no self-censoring'' in the context of block-wise missingness refers to the scenario where the missingness of the variables in a group does not depend on any of the variables within that group. The imputation results are presented in Table \ref{highdimsim}. Due to the prohibitive computational time of the MissForest method, we excluded it from these higher-dimensional simulations. However, we retained MICE as a baseline for comparison, given its ignorability assumption in contrast to other methods.

The proposed method demonstrated a clear advantage in these scenarios. The improvement of the proposed method became more pronounced as the group size increased. It is notable that with the increasing group sizes  (and consequently the decreasing number of groups) for a given number of variables, the issue of nonidentifiability is magnified in the other imputation methods assuming MNAR. Under these conditions, GINA and Not-MIWAE showed significantly deteriorating performances, whereas the proposed method exhibited relative stability. The imputation performance of MICE remained fairly stable, but its overall imputation RMSE was noticeably worse compared to the proposed method.

\begin{table}[hbt!]

\begin{center}

\caption{Imputation RMSE of simulated data with higher dimensions and $10$ latent variables, where $p$ is the number of variables in $X$ and $gs$ is the group size of the variables, i.e., the number of variables within a group.}
\vspace{3mm}
\resizebox{\columnwidth}{!}{%
\begin{tabular}{l|ccc|ccc|ccc}
\hline
\multirow{2}{*}{Method} & \multicolumn{3}{c|}{$p=30$}& \multicolumn{3}{c|}{$p=50$}&\multicolumn{3}{c}{$p=100$}\\
& $gs =2$ & $gs =3$ & $gs =5$&  $gs =2$ & $gs =5$ & $gs =10$ &  $gs =5$ & $gs =10$ & $gs =25$\\ \hline
Proposed & \B 0.61(0.02) &\B 0.64(0.02) &\B 0.74(0.16)   &  \B 0.47(0.01)  & \B 0.55(0.02)& \B 0.70(0.17) & \B 0.34(0.01) & \B0.38(0.02)& \B0.50(0.18)\\
GINA &0.70(0.04)& 0.86(0.15) &1.37(0.31) & 0.50(0.02) &  0.88(0.15) & 1.15(0.27)& 0.42(0.06) &0.60(0.09)&0.90(0.30)\\
Not-MIWAE &0.65(0.06)& 0.99(0.25)&1.79(0.31)   & 0.47(0.02) & 1.19(0.16)& 1.50(0.20) & 0.46(0.08) &0.74(0.13) &1.00(0.24)\\
MICE & 0.83(0.03) & 0.83(0.03) & 0.82(0.03) & 0.74(0.03)  & 0.75(0.02) & 0.77(0.02)& 0.60(0.01) & 0.69(0.02)&0.90(0.01)\\
\hline
\end{tabular}%
}
\label{highdimsim}
\end{center}
\vspace{-20pt}
\end{table}
\subsection{Mean estimation for mixture of Gaussian distributions} \label{mixNsim}
In this subsection, we evaluate our method from an alternative perspective, focusing on the estimation of a mean parameter. This allows us to put our model in the context of conventional statistical techniques (such as parameter estimation), as one of our goals is to \textcolor{black}{harness the strengths of these traditional methods, backed by well-established theory}.

\begin{figure}[hbt!]
    \centering
    \includegraphics[scale = 0.4]{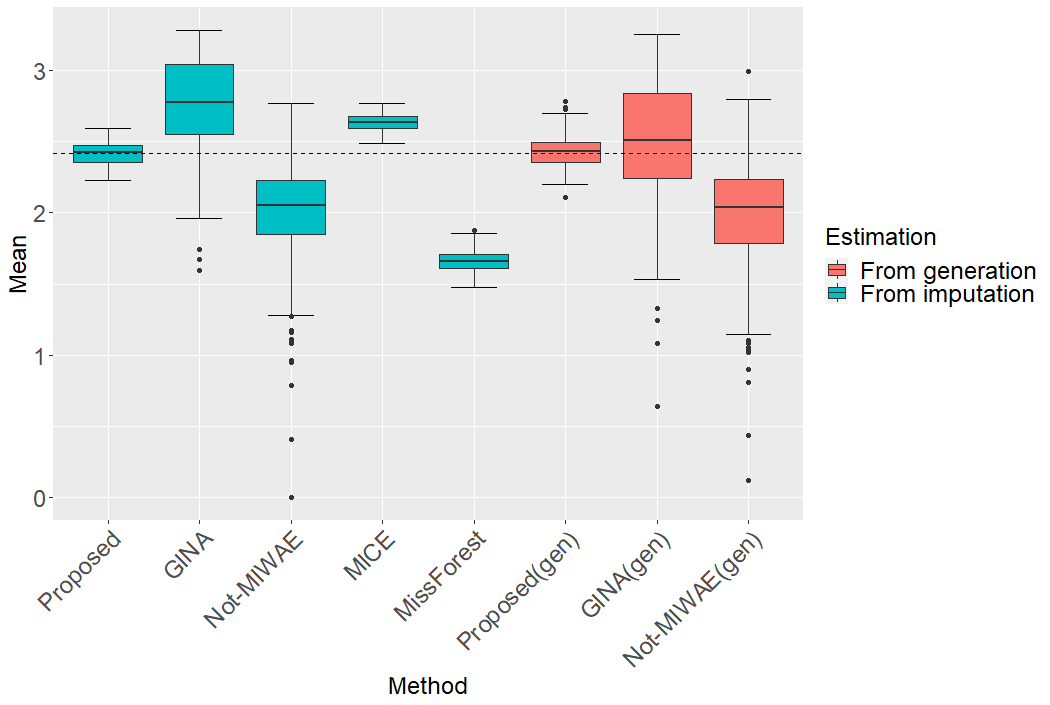}
    \caption{Estimation of $\mathbb{E}(X_3)$ in the Gaussian mixtures. The dotted line represents the true value of $\mathbb{E}(X_3)$.} 
    \label{fig:boxplot}
    \vspace{-10pt}
\end{figure}

We follow the data generation framework in \citet{malinsky2021semiparametric}, which involves a 3-dimensional mixture of Gaussian distributions with missingness introduced into the data via a conditional Gaussian graph model that complies with the no self-censoring assumption. \textcolor{black}{The missingness patterns $R = r$ are sampled from a multinomial distribution, while the data follow distinct Gaussian distributions based on varying missingness patterns, i.e. $X|R = r \sim \mathcal{N}(\mu_0(r), \Sigma_0)$. By adjusting certain interaction terms or mean parameters to zero, the no self-censoring condition can be satisfied in this parametric data generation process} \cite[p.119--120]{hojsgaard2012graphical}. The parameter values for the simulation study are detailed in the supplementary material.

In this setting,  
all variables $X = (X_1, X_2, X_3)$ could potentially be missing. Our aim is to estimate the mean of a specific variable, that is, $X_3$. We performed $100$ replications with sample size $20000$ and provided boxplots of estimated mean values of all replications and methods in Figure \ref{fig:boxplot} to illustrate the results,
where green boxes represent estimates based on observed values and imputed values, and red boxes represent estimates based on generated values only.

In particular, the estimates derived from our proposed method, whether based on imputed values or generated data, exhibit a substantial convergence with respect to the actual value represented by the dotted line in Figure \ref{fig:boxplot}. This signifies a successful recovery of the distribution, in contrast to the estimates obtained from other methods, which demonstrate biases. Our performance in mean estimation is also competitive in terms of accuracy with the semiparametric inference method proposed by \citet{malinsky2021semiparametric} for the same problem. Furthermore, the estimates generated by our approach display a limited range of variation, similar to the results of two well-established classical imputation methods operating under the inherently identifiable MAR assumption. We suspect that the large variation of the other deep generative models is due to their lack of identifiability.

\section{Real data analysis}\label{realexp}

\subsection{UCI data with synthetic missingness}

To assess the efficacy of our method in more realistic data distributions, we perform a series of experiments focusing on missing value imputation utilizing datasets from the widely recognized UCI machine learning repository \citep{asuncion2007uci}. These experiments encompass six different datasets, including banknote authentication, red wine quality, white wine quality, concrete compressive strength, yeast, and waveform database generator. These datasets 
%
%
do not contain inherent missing values. We introduce synthetic MNAR missingness 
to evaluate the performance of imputation. Since we know the truth, this allows us to compare our proposed method with other baseline techniques and state-of-the-art methods. We use the imputation RMSE to evaluate each method. 
{Considering the fact that real-world missing data with MNAR mechanisms could be different from no self-censoring, we conduct experiments using simulated MNAR mechanism which is different from no self-censoring for the investigation of robustness. In this setting, each missing indicator could potentially depend on all variables in the data.} This generation of missingness closely resembles the one in our experiments of linear no self-censoring without latent variables on synthetic data in Section \ref{sim3d}, with the exception that the self-censoring aspect is also included in the missingness mechanisms.

Specifically, we sample $U \sim \mathcal{U}(0,1)$ and
apply a linear transformation $g_{\psi_j}(X)$ on $X$, for each dimension $j$ of missingness indicators $R$, with randomized parameters $\psi_j$. 
We let $R_j=1$ if $U>g_{\psi_j}(X)$, and $R_j=0$ otherwise.
%
\begin{figure}[tbh!]
    \centering
    \includegraphics[width = \textwidth]{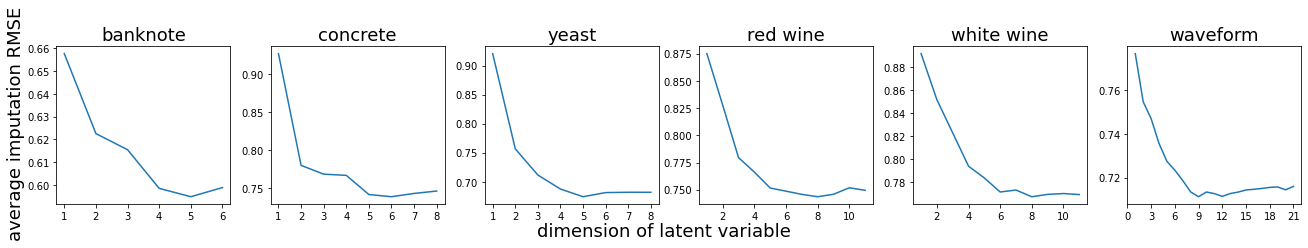}
    \caption{Selection of latent dimensions for UCI datasets with MNAR missingness. The curves represent the average imputation RMSE across each validation set during 5-fold cross-validation, where the entries to be imputed are MCAR in the validation sets.}
    \label{fig:cv}
    \vspace{-20pt}
\end{figure}

The missing rates in our experiments range from $30\%$ to $50\%$, and no variables are designed to always be observed. It is noteworthy that, unlike in simulations where latent dimensions are known, the latent dimensions in real-world data are not as clear. Consequently, we select the latent dimensions for each dataset in the presence of missingness using cross-validation, as described in Section \ref{cld}. The process of selecting the data latent dimensions involves imputing the manually MCAR missing entries in the validation sets after training the model on the training sets. Figure \ref{fig:cv} illustrates the average imputation RMSE on 5-fold cross-validated missing data with different latent dimensions for each dataset. The optimal latent dimensions were determined using the elbow method, identifying the point where the rate of improvement for imputation sharply decreases.

\begin{table}[tbh!]
\begin{center}
\caption{Imputation RMSE on  UCI Data with MNAR. The default missingness model in the generative models is linear; (nl) denotes a nonlinear missingness model and (s) represents self-censoring.}
\resizebox{\columnwidth}{!}{%
\begin{tabular}{lrrrrrr}
\hline
Method & \multicolumn{1}{c}{banknote} & \multicolumn{1}{c}{concrete}& \multicolumn{1}{c}{yeast} &  \multicolumn{1}{c}{red wine} & 
\multicolumn{1}{c}{white wine}  & \multicolumn{1}{c}{waveform}\\ \hline
Proposed & \B 0.713(0.118) & \B 0.745(0.070)&\B 0.966(0.134)&\B 0.762(0.059)&\B 0.827(0.054)&\B 0.722(0.011)\\
Proposed(nl)& \B 0.715(0.129)& \B 0.748(0.064)&\B 0.965(0.132)&\B0.787(0.055)&\B 0.840(0.056)&\B 0.730(0.013)\\
GINA & 0.843(0.106)& 0.884(0.055)&1.062(0.120)&0.877(0.051)&0.921(0.049)& 0.786(0.007)\\
GINA(nl) &0.880(0.119)& 0.926(0.056) &1.110(0.083)& 0.915(0.051)&0.970(0.041)& 0.799(0.010)\\
Not-MIWAE & 0.772(0.209)& 0.825(0.077)&1.115(0.271)&0.853(0.079)&0.870(0.059)& 0.736(0.011)\\
Not-MIWAE(nl) & 0.915(0.150)& 0.841(0.068)&1.029(0.087)&0.879(0.055)&0.901(0.041)&0.770(0.015)\\
Not-MIWAE(s) & 1.693(0.672)& 1.619(0.539)&2.225(0.710)&2.222(0.626)&2.493(0.685)&1.889(0.642)\\
MICE &1.045(0.155)& 1.054(0.061)&3.030(1.585)&1.050(0.051)&1.078(0.050)& 0.991(0.032)\\
MissForest &0.940(0.155)& 1.020(0.071)&1.534(0.289)&1.013(0.071)&1.053(0.067)&0.891(0.022)\\
\hline
\end{tabular}%
}
\label{uci_mnar}
\end{center}
\vspace{-20pt}
\end{table}
We conduct $50$ replications for simulating missingness, and explore the influence of linear and non-linear structures on deep generative models in this setting. The results in Table \ref{uci_mnar} show that
our method outperforms both other deep generative models and classical imputation baselines, even when other deep generative models employ a missingness model that matches the missingness generation in this setting{, which is likely due to their nonidentifiability}. This could {also} be attributed to the fact that the assumption of no self-censoring only excludes one variable for each dimension in the missingness mechanism, which contributes to its strong performance. The good performance may also rely on the proposed method's robustness to self-censoring, as mentioned in the supplementary material, which may result from the correlations between data variables. In contrast, Not-MIWAE(s) perform poorly in this case, as the model only considers self-censoring and ignores all other parts in the missingness mechanism, which deviates from the true model by a large amount. \textcolor{black}{Furthermore, from the comparisons of linear and nonlinear missingness models for the generative models, it is evident that the model complexity has little influence on the proposed method, which produces the best performance using either missingness model, while using the nonlinear model for GINA and Not-MIWAE yield more apparent variation in performance.} 

These experiments illustrate that our approach could prove beneficial in situations where the data is believed to exhibit MNAR patterns, but the exact missingness mechanisms are uncertain or not well-defined.
 Additionally, the finding that there is minimal discrepancy between the results obtained with linear and nonlinear choices for the missingness model in our method demonstrates that it could effectively learn the distribution, even when employing an overparameterized nonlinear structure, aligning with our theoretical findings regarding identification.

\subsection{HIV positive mothers in Botswana}
\begin{table}[t!]
\begin{center}
\caption{Estimated joint distribution of HAART continuation during pregnancy (H), low
CD$4^+$ count (C), and preterm delivery (P) in HIV-infected women in Botswana.}
\begin{tabular}{l||c|c|c|c|c|c|c|c}
\hline\hline
  P, C, H &  1, 1, 1 & 1, 1, 0 & 1, 0, 0 & 1, 0, 1 & 0, 1, 1 & 0, 1, 0 & 0, 0, 0 & 0, 0, 1\\ \hline\hline
  Proposed & 0.0188&0.0616&0.5030&0.1582&0.0112&0.0380&0.1604&0.0487\\
  AIPW & 0.0122 & 0.0268 & 0.5206 & 0.1722 & 0.0091 & 0.0542 & 0.1618 & 0.0430\\
  MICE & 0.0164 & 0.0355 & 0.5156 & 0.1705 & 0.0140 & 0.0483 & 0.1641 & 0.0357\\
  CC & 0.0137 & 0.0342 & 0.3320 & 0.0979 & 0.0333 & 0.1802 & 0.2380 & 0.0705\\
\hline\hline
\end{tabular}%
\label{t:one}
\end{center}
\vspace{-20pt}
\end{table}

A well-known dataset for the study of MNAR data is the data of HIV positive mothers in Botswana. In this dataset, the goal is to estimate the association between highly active antiretroviral therapy (HAART) during pregnancy and preterm delivery. The complete data consists of $33,148$ obstetric records from 6 cites in Botswana \citep{chen2012highly}. We follow the practice of \cite{tchetgen2018discrete, sun2018semiparametric, shpitser2016consistent, malinsky2021semiparametric} by focusing on HIV-positive
women ($n = 9711$) and $3$ binary variables: HAART exposure during pregnancy ($68.9\%$ missing), an indicator of low CD$4^+$count  (less than $200\mu L$), and preterm delivery ($6.7\%$ missing). There is a small portion of complete cases ($10.5\%$). The missingness mechanism is suspected to be MNAR \citep{shpitser2016consistent, malinsky2021semiparametric}.

We learn the distribution of the data using our method with the Bernoulli distribution for the data generative model. As the first step in estimating the relationship of interest, we obtained the joint distribution using our generative model, and compared the results with \textcolor{black}{other estimation methods} including complete case analysis (CC), MICE, and estimation using the augmented inverse probability weighted (AIPW) by \citet{malinsky2021semiparametric} which assumes a no self-censoring model. We can see from the joint distribution in Table \ref{t:one} that certain probabilities are similar for the proposed method, AIPW and MICE, \textcolor{black}{which all have substantial differences from those for the complete case analysis.}

Based on the estimated joint distribution by the proposed method, we estimate the odds ratios between HAART and preterm
delivery at different levels of CD4$^+$ count:  $1.03~ (0.53, 2.94)$ for low CD4$^+$ count and $1.04~ (0.73, 1.53)$ for moderate or high CD4$^+$ count, respectively. Without conditioning on CD4$^+$, the odds ratio is computed as $1.04~ (0.76, 1.48)$. We also estimate the effect of low CD4$^+$ on the result of delivery and obtain the odds ratio as $0.52~(0.36, 0.81)$.  The numbers in parentheses represent the 95\% confidence intervals and are computed by bootstrap (percentile) over 1000 subsamples. 
%
%
The method of AIPW found a significantly positive relationship conditional on low CD4$^+$. Our method does not find a significant result that appears similar to the result of \cite{shpitser2016consistent} with an estimation based on the pseudo-likelihood approach using a no self-censoring model.
It has been suspected that the data is sufficiently noisy to find insignificant results \citep{shpitser2016consistent}.

\subsection{Yahoo! R3 music ratings}
 In this subsection, we apply our method to the Yahoo! R3 dataset \citep{marlin2009collaborative, wang2018deconfounded}. This dataset is primarily designed for the analysis of user-song ratings, with a specific focus on evaluating MNAR imputation techniques. Within this dataset, we have a training set that consists of more than 300,000 self-selected ratings on a scale from $1$ to $5$ provided by 15,400 users for a selection of 1,000 songs, which has an entry-wise missing rate of around $98.0\%$. In addition, there is an MCAR test set, which comprises ratings from 5,400 users across a set of 10 songs chosen at random. 
 We train our model on the training set and subsequently evaluate the imputation performance on the test set. 
 \textcolor{black}{We acknowledge that this dataset likely involves self-censoring. Our aim here is not to assume otherwise, but to demonstrate the robustness and practical effectiveness of our method on a complex real-world dataset.}  

  \begin{figure}[t]
    \centering
    \includegraphics[width=1\linewidth]{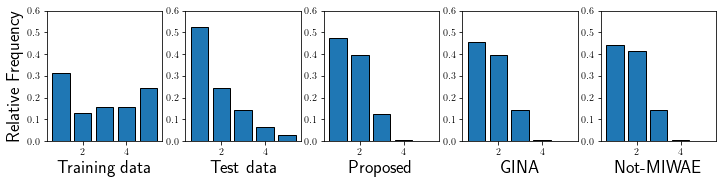}
    \caption{Bar graphs over rating values for the Yahoo! R3 dataset. The first two are from the MNAR training set and the MCAR test set. The last three show the data values of the imputed test set using different methods.}
    \label{fig:yahoo_bar}
    \vspace{-20pt}
\end{figure}
 We compare the proposed method with 
 GINA \citep{ma2021identifiable} and not-MIWAE \citep{ipsen2020not}, which are also designed for MNAR but did not address the issue of identifiability without resorting to auxiliary variables. \textcolor{black}{We model the missingness mechanism as a linear function in all methods, allowing each $R_j$ to depend on all dimensions of $X$ in the baselines.} 
 We did not include classical methods such as MICE or missForest due to their prohibitively long training time. 
  We visualize the imputation results in the five categorical bar graphs, together with the training and test data. First, we can notice that the training data have a very different distribution from the test data, where the latter is known to be a relatively unbiased representation of the true data distribution, while the training data are MNAR and thus biased. More importantly, we find that all the generative models learn a data distribution closer to the test data than to the training data, where the proposed model produces the most pronounced difference between ratings $1$ and $2$ among all the methods, \textcolor{black}{relatively closest to the distribution of the unbiased test data}.

 Since results of deep generative methods depend on choices of seeds, we re-run each deep generative method 
 for $50$ times using different seeds, yielding the results of mean and standard deviation of imputation RMSE for the different methods:  $1.1807 
 (0.0073)$, $1.1870 (0.0089)$, $1.1818 (0.0062)$ for the proposed model, GINA and not-MIWAE, respectively. 
 \textcolor{black}{Consistent with the bar graphs, our proposed method shows slightly better performance than the alternatives, though the differences are not substantial. This is expected, since in high-dimensional settings correlations across many variables may dilute the impact of self-censoring, making the no self-censoring assumption relatively mild.} 

\section{Conclusion}\label{conclusion}
In this paper, we propose a new deep latent variable model for MNAR data under a conditional no self-censoring assumption.
We have established theoretical identifiability under mild conditions. This theoretical foundation forms the basis for the novel practical algorithm IM-IWAE using importance-weighted autoencoders for data generation and imputation of MNAR data, with the theoretical justification of its ability to recover ground-truth data distributions. Our comprehensive experiments, which include simulation studies and real-world applications, effectively highlight the advantages and robustness of the proposed method. The experiments have also demonstrated that our method could be potentially useful when data is suspected to be MNAR but the specific missingness mechanisms are agnostic or unclear. 

\textcolor{black}{Some future directions include but not limited to \textcolor{black}{extending identifiable deep latent variable models for MNAR data to more general topological spaces,} enhancing the imputation performance of the proposed method in data with extremely high missing rates,
and investigating other reasonable and practical identification conditions for MNAR data to incorporate them into deep generative models.}







\appendix
 \section{Proofs of Theoretical Results}
\subsection{Proof of Theorem \ref{universal}}
\begin{proof} 
We focus on the missingness mechanism first, i.e., the conditional distribution of $R$ given $X$. 
\subsubsection{The approximation of the missingness mechanism}\label{appx_mech}
Without loss of generality, let $\tilde{Z} \sim \mathcal{N}(0,1)$, i.e. $q_2 = 1$.  (In the case of $q_2 >1$, it is trivial to show that the last $q_2 -1$ dimensions can be ignored in the conditional distribution of $R$ given $\tilde{Z}$ and $X$.) 
Then we have 
\begin{equation*}
    d H(\tilde{z}) = p(\tilde{z})d\tilde{z} = \mathcal{N}(\tilde{z}|0,1) d \tilde{z}.
\end{equation*}
Now, denote the ground-truth missingness mechanism as
\begin{equation*}
\mathbb{P}_{gt}(R = r|x) = p_r,
\end{equation*}
with $\sum_{r \in \{0,1\}^p} p_r = 1$.
With a little abuse of notation, we write below $r < r'$ by treating $r$, $r'$ as binary numbers. For example, if $r = (0,0,1)$ and $r' = (0,1,0)$, then write $r < r'.$
By using indicator functions $\bm{1}(\cdot)$, we define a set of deterministic functions 
\begin{equation*}
    f(x,u) = (f^{(1)}(x,u), \cdots, f^{(p)}(x,u))
\end{equation*}
for $x \in \mathbb{R}^p$ and $u \in [0,1]$ as 
\begin{multline*}
     f^{(1)}(x,u) =  \bm{1}\{\sum_{r' < (1,0...0)}p_{r'} < u \leq 1\}^{r_1} \bm{1}\{0 \leq u \leq \sum_{r' < (1,0...0)}p_{r'} \}^{1-r_1},
\end{multline*}
\begin{multline*}
    f^{(2)}(x,u) = \bm{1}\{\sum_{r' < (0, 1,0...0)}p_{r'} < u \leq \sum_{r' < (1, 0,0...0)}p_{r'}  \text{ or } \sum_{r' < (1, 1,0...0)}p_{r'} < u \leq 1\}^{r_2} \\ \bm{1}\{0 \leq u \leq \sum_{r' < (0, 1,0...0)}p_{r'}  \text{ or } \sum_{r' < (1, 0,0...0)}p_{r'} < u  \sum_{r' < (1, 1,0...0)}p_{r'}\}^{1-r_2}, \\
\end{multline*}
$\cdots$
\begin{multline*}
    f^{(p)}(x,u) = \bm{1}\{\sum_{r' < (0, ...,0, 1)}p_{r'} < u \leq \sum_{r' < (0, ..., 0, 1, 0)}p_{r'}  \text{ or } \cdots \text{ or } \sum_{r' < (1,..., 1)}p_{r'} < u \leq 1\}^{r_p} \\\bm{1}\{0 \leq u \leq \sum_{r' < (0, ..., 0, 1)}p_{r'}  \text{ or } \sum_{r' < (1, ..., 1, 0)}p_{r'} < u  \sum_{r' < (1, ..., 1)}p_{r'}\}^{1-r_p}.
\end{multline*}
Then we have that 
\begin{equation}
\begin{aligned}
    \mathbb{P}_{gt}(R = r|x) &= \int_0^1 \prod^p_{j=1} \text{Bernoulli}(f^{(j)}(x,u)) d u\\
    &= \int_0^1 \prod^p_{j=1} \text{Bernoulli}(f^{(j)}(x,H(\tilde{z})) d H(\tilde{z})\\
    &= \int_{\mathbb{R}} \prod^p_{j=1} \text{Bernoulli}(f^{(j)}(x,H(\tilde{z})) p(\tilde{z})d\tilde{z}\label{prgx}
\end{aligned}
\end{equation}
Hence, there exists a decomposition of the missingness mechanism in the form of 
\begin{equation*}
     \mathbb{P}_{gt}(R = r|x) = \int_{\mathbb{R}} \mathbb{P}_{gt}(R = r|x, \tilde{z})p(\tilde{z})d\tilde{z}, 
\end{equation*}
where 
\begin{equation*}
     \mathbb{P}_{gt}(R = r|x, \tilde{z}) = \prod^p_{j=1} \text{Bernoulli}(f^{(j)}(x,H(\tilde{z}))).
\end{equation*}

For the part of the data variables $X$, we separate the proof into two cases $d<p$ and $d=p$, where $d$ is the dimension of the manifold $\mathcal{X}$ for the data variable $X$.

\subsubsection{The approximation for the data distribution when $d < p$}\label{thm1pf_lowd} 
We first examine the  variable $Z$ for the data $X$. For clarity, let $p^x_{gt}(\cdot)$ denote the ground-truth manifold density of $X$(instead of simply $p_{gt}(\cdot))$. We aim to construct a bijection between $\mathcal{X}$ and $\mathbb{R}^d$ which transforms the ground-truth measure $\mu_{gt}$ to a Gaussian distribution, before which one step is needed that bijects between $\mathcal{X}$ and $\mathbb{R}^d$ using a diffeomorphism $\varphi(\cdot)$, so it transforms the ground-truth probability distribution $p^x_{gt}(x)$ to another distribution $p^u_{gt}(u)$, where $u \in \mathbb{R}^d$. The relationship between the two distributions is
\begin{equation*}
    p^u_{gt}(u)du = p^x_{gt}(x) \mu_V(dx)|_{x = \varphi^{-1}(u)} = \mu_{gt}(dx),
\end{equation*}
where $\mu_V(dx)$ is the volume measure with respect to $\mathcal{X}$. 

Define the function $F: \mathbb{R}^d \mapsto [0,1]^{d}$ as 
    \begin{equation*}
        F(u) = [F_1(u_{1}),F_2(u_{2};u_{1}), \cdots,F_{d}(u_{d};u_{d-1})],
    \end{equation*}
    with
    \begin{equation*}
        F_j(u_{j};u_{1:j-1}) = \int^{u_{ j}}_{u'_{j} = -\infty} p^u_{gt}(u'_{j}|u_{1:j-1})du'_{j}.
    \end{equation*}
    By this definition, we have 
    \begin{equation*}
        d F(u) = p^u_{gt}(u) d u.
    \end{equation*}
    Because $\varphi$ is a diffeomorphism, both $\varphi$ and $\varphi^{-1}$ are differentiable. Thus $dx/du$ is nonzero everywhere on the manifold. Considering $p^x_{gt}(x)$ is also nonzero everywhere, $p^u_{gt}(u)$ is nonzero everywhere.

Since $p^u_{gt}(u)$ is nonzero everywhere, $F(\cdot)$ is invertible. 
    W.l.o.g., let $Z \sim \mathcal{N}(0,I_{q_1})$ and $q_1 = d$. (In the case of $q_1 >d$, it is trivial to show that the last $q_1 -d$ dimensions can be ignored in the conditional distribution of $X$ given $Z$.) We define another differentiable and invertible function $G: \mathbb{R}^{d} \mapsto [0,1]^{d}$ as 
    \begin{equation*}
        G(z) = [G_1(z_{1}), G_2(z_{2}), \cdots, G_{d}(z_{{d}})]^T,
    \end{equation*}
    with
    \begin{equation*}
        G_j(z_{j}) = \int^{z_{j}}_{z'_{j} = -\infty}  \mathcal{N}(z_{j}|0,1)dz'_{j}.
    \end{equation*}
    Then 
    \begin{equation*}
        dG(z) = p(z)dz = \mathcal{N}(z|0,I)dz.
    \end{equation*}
    
    Now, consider a random variable $Y \in \mathbb{R}^p$ with a parameter $\beta > 0$ for its density: 
    \begin{equation*}
        p(y;\beta) = \int_{\mathbb{R}^d} p(y|z;\beta) p(z)dz.
    \end{equation*}
    Let 
    \begin{equation*}
        p(y|z;\beta) = \mathcal{N}(y|\varphi^{-1} \circ F^{-1} \circ G(z), \frac{1}{\beta}I)
    \end{equation*}
    so that 
    \begin{equation}
    \begin{aligned}
        p(y;\beta) &= \int_{\mathbb{R}^d} \mathcal{N}(y|\varphi^{-1} \circ F^{-1} \circ G(z), \frac{1}{\beta}I) dG(z)\\
        &= \int_{[0,1]^d} \mathcal{N}(y|\varphi^{-1} \circ F^{-1}(\xi), \frac{1}{\beta}I)d\xi\\
        &=   \int_{\mathbb{R}^d} \mathcal{N}(y|\varphi^{-1}(u),\frac{1}{\beta}I) p^u_{gt}(u) du\\
        &= \int_{y' \in \mathcal{X}} \mathcal{N}(y|y',\frac{1}{\beta}I) \mu_{gt}(dy')\label{py_beta}
    \end{aligned}
    \end{equation}

    Consider a measurable set $A \in \mathbb{R}^p$. Let 
    \begin{equation*}
        p_{R|X,\tilde{Z}}(r|y,\tilde{z}) =  \prod^p_{j=1} \text{Bernoulli}(f^{(j)}(y,H(\tilde{z})) 
    \end{equation*}
    as in (\ref{prgx}). Also, by (\ref{prgx}), we have 
    \begin{equation}\label{equal_gt}
        \lim_{\beta \rightarrow \infty} \int_{y \in A} \int \int p(y|z;\beta) p(z) p_{R|X,\tilde{Z}}(r|y,\tilde{z}) p(\tilde{z}) dz d\tilde{z} dy =\lim_{\beta \rightarrow \infty} \int_{y \in A}  p(y|\beta) p_{gt}(r|y) dy.
    \end{equation}
    Furthermore,
    \begin{equation}\label{yinA}
        \begin{aligned}
        \lim_{\beta \rightarrow \infty} \int_{y \in A}  p(y|\beta) p_{gt}(r|y) dy 
        &= \lim_{\beta \rightarrow \infty} \int_{y \in A}\left[ \int_{y' \in \mathcal{X}} \mathcal{N}(y|y',\frac{1}{\beta}I) \mu_{gt}(dy') p_{gt}(r|y)\right]dy\\
        &=  \lim_{\beta \rightarrow \infty}\int_{y' \in \mathcal{X}}\left[  \int_{y \in A}  \mathcal{N}(y|y',\frac{1}{\beta}I)   p_{gt}(r|y)dy\right]\mu_{gt}(dy')\\
        &= \int_{y' \in \mathcal{X}} \lim_{\beta \rightarrow \infty} \left[  \int_{y \in A}  \mathcal{N}(y|y',\frac{1}{\beta}I)   p_{gt}(r|y)dy\right]\mu_{gt}(dy').
    \end{aligned}
    \end{equation}
    The second equality follows by Fubini’s theorem. The third equality follows by the bounded convergence theorem. We now note that the term inside the first integration. Now, note that
    \[
    \lim_{\beta \rightarrow \infty} \left[  \int_{y \in A}  \mathcal{N}(y|y',\frac{1}{\beta}I) p_{gt}(r|y) dy \right] = 
    \begin{cases} 
        p_{gt}(r|y) & \text{if } y' \in A - \partial A, \\
        0 & \text{if } y' \in A^c - \partial A.
    \end{cases}
    \]
    We separate the manifold $\mathcal{X}$ into three parts: $\mathcal{X} \cap (A - \partial A)$, $\mathcal{X} \cap (A^c - \partial A)$ and $\mathcal{X} \cap  \partial A$. Then (\ref{yinA}) can be separated into three parts correspondingly. The first two parts can be derived as
    \begin{align*}
        \int_{\mathcal{X}\cap (A - \partial A)} \lim_{\beta \rightarrow \infty} \left[  \int_{y \in A}  \mathcal{N}(y|y',\frac{1}{\beta}I)   p_{gt}(r|y)dy\right]\mu_{gt}(dy') = \int_{\mathcal{X}\cap (A - \partial A)} p_{gt}(r|y) \mu_{gt}(dy') ,
    \end{align*}
    \begin{align*}
        \int_{\mathcal{X}\cap (A^c - \partial A)} \lim_{\beta \rightarrow \infty} \left[  \int_{y \in A}  \mathcal{N}(y|y',\frac{1}{\beta}I)   p_{gt}(r|y)dy\right]\mu_{gt}(dy') = \int_{\mathcal{X}\cap (A - \partial A)} 0\mu_{gt}(dy') = 0.
    \end{align*}
    For the third part, given the assumption that $\mu_{gt}(\partial A) = 0$, we have 
    \begin{equation*}
        0 \leq \int_{\mathcal{X}\cap \partial A} \lim_{\beta \rightarrow \infty} \left[  \int_{y \in A}  \mathcal{N}(y|y',\frac{1}{\beta}I)   p_{gt}(r|y)dy\right]\mu_{gt}(dy')\leq \int_{\mathcal{X}\cap \partial A}1 \mu_{gt}(dy')= \mu_{gt}(\mathcal{X}\cap \partial A) = 0.
    \end{equation*}
    Therefore, we have
    \begin{equation*}
        \int_{\mathcal{X}\cap \partial A} \lim_{\beta \rightarrow \infty} \left[  \int_{y \in A}  \mathcal{N}(y|y',\frac{1}{\beta}I)   p_{gt}(r|y)dy\right]\mu_{gt}(dy') = 0
    \end{equation*}
    and so
    \begin{align*}
        \lim_{\beta \rightarrow \infty} \int_{y \in A}  p(y|\beta) p_{gt}(r|y) dy &=  \int_{\mathcal{X}} \lim_{\beta \rightarrow \infty} \left[  \int_{y \in A}  \mathcal{N}(y|y',\frac{1}{\beta}I)   p_{gt}(r|y)dy\right]\mu_{gt}(dy') \\
        & = \int_{\mathcal{X}\cap (A - \partial A)} p_{gt}(r|y) \mu_{gt}(dy') + 0 + 0\\
        & = \int_{\mathcal{X}\cap A} p_{gt}(r|y) \mu_{gt}(dy').
    \end{align*}
    Combining this and (\ref{equal_gt}), we have the result as desired.

\subsubsection{The approximation for the data distribution when $d = p$}\label{thm1pf_dp}
Take $\varphi$ in the case of $d=p$ to be the identity function (so that we can actually omit the diffeomorphism step when constructing a bijection between $\mathcal{X}$ and $\mathbb{R}^d$.)
        In the case of $d=p$, $p^x_{gt}(x)\triangleq \mu_{gt}(dx)/dx$ represents the ground-truth probability density with respect to the standard Lebesgue measure in Euclidean space. 
    Define the function $F: \mathbb{R}^p \mapsto [0,1]^{p}$ as
    \begin{equation*}
        F(x) = [F_1(x_{1}),F_2(x_{2};x_{1}), \cdots,F_{p}(x_{p};x_{p-1})],
    \end{equation*}
    \begin{equation*}
        F_j(x_{j};x_{1:j-1}) = \int^{x_{ j}}_{x'_{j} = -\infty} p^x_{gt}(x'_{j}|x_{1:j-1})dx'_{j}.
    \end{equation*}
    By this definition, we have 
    \begin{equation*}
        d F(x) = p^x_{gt}(x) d x.
    \end{equation*}
    Since $p^x_{gt}(x)$ is nonzero everywhere, $F(\cdot)$ is invertible. 
    W.l.o.g., let $Z \sim \mathcal{N}(0,I_{q_1})$ and $q_1 = p$. (In the case of $q_1 >p$, it is trivial to show that the last $q_1 -p$ dimensions can be ignored in the conditional distribution of $X$ given $Z$.) We define another differentiable and invertible function $G: \mathbb{R}^{p} \mapsto [0,1]^{p}$ as
    \begin{equation*}
        G(z) = [G_1(z_{1}), G_2(z_{2}), \cdots, G_{p}(z_{{p}})]^T,
    \end{equation*}
    \begin{equation*}
        G_j(z_{j}) = \int^{z_{j}}_{z'_{j} = -\infty}  \mathcal{N}(z_{j}|0,1)dz'_{j}.
    \end{equation*}
    Then
    \begin{equation*}
        dG(z) = p(z)dz = \mathcal{N}(z|0,I)dz.
    \end{equation*}
    Now, consider a random variable $Y \in \mathbb{R}^p$ with a parameter $\beta > 0$ for its density: 
    \begin{equation*}
        p(y;\beta) = \int_{\mathbb{R}^p} p(y|z;\beta) p(z)dz.
    \end{equation*}
    Let 
    \begin{equation*}
        p(y|z;\beta) = \mathcal{N}(y|F^{-1} \circ G(z), \frac{1}{\beta}I)
    \end{equation*}
    so that 
    \begin{align*}
        p(y;\beta) &= \int_{\mathbb{R}^p} \mathcal{N}(y|F^{-1} \circ G(z), \frac{1}{\beta}I) dG(z)\\
        &= \int_{[0,1]^p} \mathcal{N}(y|F^{-1}(\xi), \frac{1}{\beta}I)d\xi,
    \end{align*}
    where $\xi = G(z)$.
    Let $y' = F^{-1} (\xi)$ such that $d\xi = dF(y') = p^x_{gt}(y')dy'$. Plugging this expression into $p(y;\beta)$ we have
    \begin{equation*}
        p(y;\beta) = \int_{\mathbb{R}^p} \mathcal{N}(y|y',\frac{1}{\beta}I) p^x_{gt}(y') dy'.
    \end{equation*}
    As $\beta \rightarrow \infty$, $\mathcal{N}(y|y',\frac{1}{\beta}I)$ becomes a Dirac-delta function, resulting in 
    \begin{equation*}
        \lim_{\beta \rightarrow \infty} p(y;\beta) = \int_{\mathbb{R}^p} \delta(y' - y) p^x_{gt}(y')dy' = p^x_{gt}(y).
    \end{equation*}

Putting the results for $p_{gt}(x)$ and $p_{gt}(r|x)$ together, we have the result for the case $d=p$ in the theorem.

\end{proof}
\subsection{Proof of Theorem \ref{mechanism}}
\begin{proof}
Under Assumption \ref{nscass},
\begin{equation*}
\begin{aligned}
    p_{gt}(r|x) & = \int p_{gt}(r|x, \tilde{z}) p(\tilde{z}) d\tilde{z}\\
    & = \int \prod^p_{j=1} p_{gt}(r_j|x_{-j}, \tilde{z}) p(\tilde{z}) d\tilde{z}\\
    & = \int \prod^p_{j=1} p_{gt}(r_j|x_{-j}, R_{-j} = \bm{1}, \tilde{z}) p(\tilde{z}) d\tilde{z}.
\end{aligned}
\end{equation*}
Thus, $p_{gt}(r|x)$ is only a function of observed data.
\end{proof}
\subsection{Proof of Corollary \ref{cor1}}%
\begin{proof}
    
Using the odds ratio parametrization \citep{yun2007semiparametric, chen2010compatibility} with odds ratio function
\begin{equation*}
    \textrm{OR}(r,x) \equiv \textrm{OR}(r,x;r_0 = \bm{1}, x_0 = \bm{0}) = \frac{p_{gt}(r|x)}{p_{gt}(R=\bm{1}|x)}\frac{p_{gt}(R=\bm{1}|X=\bm{0})}{p_{gt}(r|X=\bm{0})}.
\end{equation*}
Then the full-data distribution \citep{li2022self} can be written as
\begin{equation*}
    p_{gt}(x,r) = \frac{OR(r,x)p_{gt}(x|R=\bm{1})p_{gt}(r|X = \bm{0})}{\sum_{r'}\mathbb{E}[OR(r',y)|R=\bm{1}]p_{gt}(r'|X=\bm{0})},
\end{equation*}
which is also a function of observed data only.
\end{proof}

\subsection{Proof of Theorem \ref{LK}}
\begin{proof}
We need to show $\log p_{\theta, \psi}(x_{(r)},r) \ge L_k$, $L_k \geq L_m$ for $k \geq m \geq 1$; we also need to show $\log p_{\theta, \psi}(x_{(r)},r)  = \lim_{k \rightarrow \infty} L_k$ assuming $w$ is bounded, with 
    \begin{equation*}
        \mathcal{L}_k (\theta, \psi, \phi, \lambda) = \mathbb{E} \big[\log \frac{1}{k} \sum^k_{l=1} w_k\big],
    \end{equation*}
    where 
    \begin{equation*}
        w_l =\frac{p_{\theta}(x_{(r)}|z_l) p_{\psi}(r|x_{(r)}, x_{(m),l}, \tilde{z_l})p(z_l)p(\tilde{z_l})}{q_{\phi}(z_l|x_{(r)}, r)q_{\lambda}(\tilde{z_l}|x_{(r)}, r)},
    \end{equation*}
    for $1 \leq l \leq k.$
We have
\begin{equation*}
        \mathcal{L}_k (\theta, \psi, \phi, \lambda) = \mathbb{E} \big[\log \frac{1}{k} \sum^k_{l=1} w_l\big] \leq \log \mathbb{E} \big[ \frac{1}{k} \sum^k_{l=1} w_l\big] = \log p_{\theta, \psi}(x_{(r)},r).
\end{equation*}
For $k \geq m$, denoting $h_l = (z_l, \tilde{z_l}, x_{(m),l})$ for $1 \leq l \leq k$
\begin{equation*}
\begin{aligned}
    \mathcal{L}_k  & = \mathbb{E}_{h_1,\cdots, h_k} \big[\log \frac{1}{k} \sum^k_{l=1} w_l\big]= \mathbb{E}_{h_1,\cdots, h_k} \big[\log \mathbb{E}_{I = \{l_1,...,l_m\}}\big[\frac{1}{m} \sum^m_{j=1} w_{l_j}\big]\big]\\
    & \geq \mathbb{E}_{h_1,\cdots, h_k} \big[\mathbb{E}_{I = \{l_1,...,l_m\}}\big[\log \frac{1}{m} \sum^m_{j=1} w_{l_j}\big]\big]= \mathbb{E}_{h_1,\cdots, h_m} \big[\log \frac{1}{m} \sum^m_{l=1} w_l\big] = \mathcal{L}_m.
\end{aligned}
\end{equation*}
Consider the random variable $M_k = \frac{1}{k} \sum^k_{l=1} w_l$. If $w$ is bounded. then it follows from SLLN that $M_k$ converges to $\mathbb{E}_{q(h|x_{(r)},r)}[w] = p_{\theta, \psi}(x_{(r)},r)$ almost surely. Hence $\mathcal{L}_k = \mathbb{E} [\log M_k]$ converges to $ \log p_{\theta, \psi}(x_{(r)},r)$ as $k \rightarrow \infty$.
\end{proof}
\subsection{Proof of Proposition \ref{prop1}}
\begin{proof}
 First, define the random variable $Y_K \coloneqq \frac{1}{K}\sum^K_{k=1} w_k$, which corresponds to the sample mean of $w_1, \cdots, w_K.$ We have
\begin{align*}
    \hat{\mathcal{L}}_K = \log Y_K & = \log(\mathbb{E}[w] + (Y_K - \mathbb{E}[w] ))=  \log \mathbb{E}[w] - \sum^{\infty}_{i=1} \frac{(-1)^i}{i\mathbb{E}[w]^i} (Y_K - \mathbb{E}[w])^i ,
\end{align*}
and so
\begin{equation}\label{taylorL}
    \mathbb{E}[\hat{\mathcal{L}}_K] = \mathbb{E}[\log Y_K] = \log \mathbb{E}[w] -\sum^{\infty}_{i=1} \frac{(-1)^i}{i\mathbb{E}[w]^i} \mathbb{E}[(Y_K - \mathbb{E}[w])^i].
\end{equation}

Denote $\gamma_i \coloneqq \mathbb{E}[(Y_K - \mathbb{E}[Y_K])^i ]$, the $i$-th central moments of $Y_K$ for $i \geq 2$, and $\mu_i \coloneqq \mathbb{E}[(w - \mathbb{E}[w])^i ]$, the $i$-th central moments of $P$ for $i \geq 2$. Denote $\gamma = \mathbb{E}[Y_K]$ and $\mu = \mathbb{E}[w]$, the first non-central moments. Since $Y_K$ is a sample mean, we can use existing results that relate $\gamma_i$ to $\mu_i$. In \cite{angelova2012moments}, Theorem 1 gives the relations: $\gamma = \mu,
    \gamma_2 = \frac{\mu_2}{K},
    \gamma_3 = \frac{\mu_3}{K^2},
    \gamma_4 = \frac{3}{K^2}\mu^2_2 + \frac{1}{K^3}(\mu_4 - 3\mu^2_2),
    \gamma_5 = \frac{10}{K^3}\mu_3\mu_2 + \frac{1}{K^4}(\mu_5 - 10 \mu_3 \mu_2),$
where the first two are well-known relations of means and variances of between samples and populations. Expanding (\ref{taylorL}) to the fifth order using the above relations gives
\begin{multline}\label{ELK}
    \mathbb{E}[\hat{\mathcal{L}}_K] = \log \mathbb{E}[w] -\frac{1}{2\mu^2}\frac{\mu_2}{K} + \frac{1}{3\mu^3}\frac{\mu_3}{K^2} - \frac{1}{4\mu^4}\left(\frac{3}{K^2}\mu^2_2+\frac{1}{K^3}(\mu_4 - 3\mu^2_2)\right) \\
    +\frac{1}{5\mu^5}\left(\frac{10}{K^3}\mu_3\mu_2 + \frac{1}{K^4}(\mu_5 - 10\mu_3\mu_2)\right) + o(K^{-3}).
\end{multline}
The result for bias follows by regrouping the terms by order of $K$ and then subtracting $\log p_{\theta, \psi}(x_{(r)},r) = \log \mathbb{E}[w]$ from (\ref{ELK}).

By using the definition of the variance and the series expansion of
the logarithm function, we have
\begin{align*}
    \text{Var}(\log Y_K) &= \mathbb{E}[(\log Y_K - \mathbb{E}[\log Y_K])^2]\\
    &=  \mathbb{E}\left[\left(\log \mu - \sum^{\infty}_{i=1} \frac{(-1)^i}{i\mu^i} (Y_K - \mu)^i - \log \mu +\sum^{\infty}_{i=1} \frac{(-1)^i}{i\mu^i} \mathbb{E}[(Y_K - \mu)^i] \right)^2\right]\\
    &= \mathbb{E}\left[\left( \sum^{\infty}_{i=1} \frac{(-1)^i}{i\mu^i} (\mathbb{E}[(Y_K - \mu)^i]-  (Y_K - \mu)^i ) \right)^2\right].
\end{align*}
By expanding the above equation to the third order we have the following:
\begin{align*}
    \text{Var}(\log Y_K) \approx& \frac{\gamma_2}{\mu^2} - \frac{1}{\mu^3}(\gamma_3 - \gamma_1\gamma_2)+\frac{2}{3\mu^4}(\gamma_4 - \gamma_1\gamma_3) +\frac{1}{4\mu^4}(\gamma_4 - \gamma^2_2) \\   & -\frac{1}{3\mu^5}(\gamma_5 - \gamma_2 \gamma_3) + \frac{1}{9\mu^6}(\gamma_6 - \gamma_3^2).
\end{align*}
By substituting the sample moments $\gamma_i$ of $Y_K$ with the central moments $\mu_i$ of distribution $P$ using their relations, we get
\begin{equation*}
    \text{Var}(\log Y_K) = \frac{1}{K}\frac{\mu_2}{\mu^2} - \frac{1}{K^2} \left( \frac{\mu_3}{\mu^3} - \frac{5\mu^2_2}{2\mu^4}\right) + o(K^{-2}).
\end{equation*}
\end{proof}%
\subsection{Proof of Corollary \ref{consistency}}
\begin{proof}\begin{align*}
        \mathbb{P}(|\hat{\mathcal{L}}_K - \log p_{\theta, \psi}(x_{(r)},r) | \geq \epsilon) &= \mathbb{P}(|\hat{\mathcal{L}}_K - \mathbb{E}[\hat{\mathcal{L}}_K] + \mathbb{E}[\hat{\mathcal{L}}_K]\log p_{\theta, \psi}(x_{(r)},r) | \geq \epsilon)\\
        & \leq \mathbb{P}(|\hat{\mathcal{L}}_K - \mathbb{E}[\hat{\mathcal{L}}_K]| + |\mathbb{E}[\hat{\mathcal{L}}_K]\log p_{\theta, \psi}(x_{(r)},r) | \geq \epsilon) \\
        & \leq \mathbb{P}(|\hat{\mathcal{L}}_K - \mathbb{E}[\hat{\mathcal{L}}_K]| \geq \frac{\epsilon}{2}) + \mathbb{P}(|\mathbb{E}[\hat{\mathcal{L}}_K]\log p_{\theta, \psi}(x_{(r)},r) | \geq \frac{\epsilon}{2}),
    \end{align*}
    where the second term is deterministic and is either $0$ or $1$, and it will always be zero for $K$ large enough according to Proposition 1. Applying Chebyshev's inequality to the first term and setting $\tau = \frac{\epsilon}{2\sqrt{\text{Var}(\hat{\mathcal{L}}_K)}}$, we have 
    \begin{align*}
        \mathbb{P}(|\hat{\mathcal{L}}_K - \mathbb{E}[\hat{\mathcal{L}}_K]| \geq \frac{\epsilon}{2}) &= \mathbb{P}(|\hat{\mathcal{L}}_K - \mathbb{E}[\hat{\mathcal{L}}_K]| \geq \tau \sqrt{\text{Var}(\hat{\mathcal{L}}_K)})\leq \frac{1}{\tau^2}
        = 4\frac{ \sqrt{\text{Var}(\hat{\mathcal{L}}_K)}}{\epsilon^2}
        = O(1/K).
    \end{align*}
    For $K \to \infty$ and any $\epsilon >0$, this term has a limit of 0. And we have the convergence in probability, and hence consistency.
\end{proof}

\subsection{Proof of Lemma \ref{lemma_r}}%
\begin{proof}
Without loss of generality, we assume the dimension of $\tilde{z}$ to be $\kappa_2 = 1$ in the proof. (In the case of $\kappa_2 >1$, it is trivial to show that the last $\kappa_2 -1$ dimensions can be ignored in the conditional distribution of $R$ given $\tilde{Z}$ and $X$.) 
Define a differentiable and invertible function $G: \mathbb{R}\rightarrow [0,1]$ as 
\begin{equation*}
G(\tilde{z}) = \int^{\tilde{z}}_{\tilde{z}' = -\infty} \mathcal{N}(\tilde{z}|0,1) d \tilde{z}'.
\end{equation*}
Then 
\begin{equation*}
    d G(\tilde{z}) = p(\tilde{z})d\tilde{z} = \mathcal{N}(\tilde{z}|0,1) d \tilde{z}.
\end{equation*}
For $r_j \in \{0,1\}$ for $1 \leq j \leq p$, we have
\begin{align*}
        \mathbb{P}_{\psi} (R_1 = r_1, \cdots, R_p = r_p|x) & = \int_{\mathbb{R}} \prod^p_{j=1} \mathbb{P}_{\psi}(R_j = r_j|x,\tilde{z})p(\tilde{z}) d \tilde{z} \\
        & = \int_{\mathbb{R}} \prod^p_{j=1} \mathbb{P}_{\psi}(R_j = r_j|x,G(\tilde{z})) d G(\tilde{z})\\
        & = \int_0^1 \prod^p_{j=1} \mathbb{P}_{\psi}(R_j = r_j|x,u) d u\\
        & = \int_0^1 \prod^p_{j=1} \text{Bernoulli}(f^{(j)}_{\psi}(x,u)) d u,
\end{align*}
where $f_{\psi}(x,u) = (f^{(1)}_{\psi}(x,u), \cdots, f^{(p)}_{\psi}(x,u))$ is the deterministic decoder for the missingness mechanism.
Now, denote the ground-truth joint density as 
\begin{equation*}
\mathbb{P}_{gt}(R = r|x) = p_r, \quad \text{with } \sum_{r \in \{0,1\}^p} p_r = 1.
\end{equation*}

With a little abuse of notation, we write below $r < r'$ by treating $r$, $r'$ as binary numbers. For example, if $r = (0,0,1)$ and $r' = (0,1,0)$, then write $r < r'.$
By using indicator functions $\bm{1}(\cdot)$, we can write the ground-truth distribution of the missingness mechanism as
\begin{multline*}
     \mathbb{P}_{gt}(R = r|x) = \int_{[0,1]} \bm{1}\{\sum_{r' < (1,0...0)}p_{r'} < u \leq 1\}^{r_1} \bm{1}\{0 \leq u \leq \sum_{r' < (1,0...0)}p_{r'} \}^{1-r_1} \\
      \bm{1}\{\sum_{r' < (0, 1,0...0)}p_{r'} < u \leq \sum_{r' < (1, 0,0...0)}p_{r'}  \text{ or } \sum_{r' < (1, 1,0...0)}p_{r'} < u \leq 1\}^{r_2} \\ \bm{1}\{0 \leq u \leq \sum_{r' < (0, 1,0...0)}p_{r'}  \text{ or } \sum_{r' < (1, 0,0...0)}p_{r'} < u  \sum_{r' < (1, 1,0...0)}p_{r'}\}^{1-r_2} \\
      \cdots  \\
      \bm{1}\{\sum_{r' < (0, ...,0, 1)}p_{r'} < u \leq \sum_{r' < (0, ..., 0, 1, 0)}p_{r'}  \text{ or } \cdots \text{ or } \sum_{r' < (1,..., 1)}p_{r'} < u \leq 1\}^{r_p} \\ \bm{1}\{0 \leq u \leq \sum_{r' < (0, ..., 0, 1)}p_{r'}  \text{ or } \sum_{r' < (1, ..., 1, 0)}p_{r'} < u  \sum_{r' < (1, ..., 1)}p_{r'}\}^{1-r_p}
 d u.
\end{multline*}
Hence, we can see that by setting 
\begin{equation*}
     f^{(1)}_{\psi^*}(x,u) =  \bm{1}\{\sum_{r' < (1,0...0)}p_{r'} < u \leq 1\}^{r_1} \bm{1}\{0 \leq u \leq \sum_{r' < (1,0...0)}p_{r'} \}^{1-r_1},
\end{equation*}
\begin{multline*}
    f^{(2)}_{\psi^*}(x,u) = \bm{1}\{\sum_{r' < (0, 1,0...0)}p_{r'} < u \leq \sum_{r' < (1, 0,0...0)}p_{r'}  \text{ or } \sum_{r' < (1, 1,0...0)}p_{r'} < u \leq 1\}^{r_2} \\ \bm{1}\{0 \leq u \leq \sum_{r' < (0, 1,0...0)}p_{r'}  \text{ or } \sum_{r' < (1, 0,0...0)}p_{r'} < u  \sum_{r' < (1, 1,0...0)}p_{r'}\}^{1-r_2}, \\
\end{multline*}
$\cdots$
\begin{multline*}
    f^{(p)}_{\psi^*}(x,u) = \bm{1}\{\sum_{r' < (0, ...,0, 1)}p_{r'} < u \leq \sum_{r' < (0, ..., 0, 1, 0)}p_{r'}  \text{ or } \cdots \text{ or } \sum_{r' < (1,..., 1)}p_{r'} < u \leq 1\}^{r_p} \\ \bm{1}\{0 \leq u \leq \sum_{r' < (0, ..., 0, 1)}p_{r'}  \text{ or } \sum_{r' < (1, ..., 1, 0)}p_{r'} < u  \sum_{r' < (1, ..., 1)}p_{r'}\}^{1-r_p},
\end{multline*}
the decoder $f_{\psi^*}(x,u)$ can recover the ground-truth distribution of missingness mechanism, i.e., 
\begin{equation*}
    \mathbb{P}_{gt}(R = r|x) = \int_0^1 \prod^p_{j=1} \text{Bernoulli}(f^{(j)}_{\psi^*}(x,u)) d u.
\end{equation*}
\end{proof}
\subsection{Proof of Theorem \ref{final_thm}}%
\begin{proof}
    For VAE (the special case of $K=1$ for IWAE) for complete data, \cite{dai2019diagnosing} showed that there is a sequence of decoders such that 
    \begin{equation*}
        \lim_{t \rightarrow \infty} p_{\theta^*_t}(x) = p_{gt}(x),
    \end{equation*}
    under certain conditions.
    We first show that such a result on complete data is applicable to the observed data when there is missingness in the data.
    
    Denote the dimension of the observed data as $p_{(r)}$ for $1 \leq p_{(r)} \leq p$, with $p_{(r)}$ being the number of $1$'s in $r$.
    Define the function $F: \mathbb{R}^{p_{(r)}} \mapsto [0,1]^{p_{(r)}}$ as
    \begin{equation*}
        F(x_{(r)}) = [F_1(x_{(r),1}),F_2(x_{(r),2};x_{(r),1}), \cdots,F_{p_{(r)}}(x_{(r),p_{(r)}};x_{(r),p_{(r)}-1})],
    \end{equation*}
    where for $1\leq j\leq p_{(r)}$, $F_j$ is defined as
    \begin{equation*}
        F_j(x_{(r), j};x_{(r),1:j-1}) = \int^{x_{(r), j}}_{x'_{(r), j} = -\infty} p_{gt}(x'_{(r), j}|x_{(r), 1:j-1})dx'_{(r), j},
    \end{equation*}
    where $x_{(r),j}$ denotes the $j^\text{th}$ observed dimension in $x_{(r)}$. 
    By this definition, we have 
    \begin{equation*}
        d F(x_{(r)}) = p_{gt}(x_{(r)}) d x_{(r)}.
    \end{equation*}
    Since $p_{gt}(x)$, and hence $p_{gt}(x_{(r)})$, is assumed to be nonzero everywhere, $F(\cdot)$ is invertible. 
    
    Similarly, we define another differentiable and invertible function $G: \mathbb{R}^{p_{(r)}} \mapsto [0,1]^{p_{(r)}}$ for the part of the latent variable corresponding to the observed part of the data $x_{(r)}$, which we denote as $z_{(r)}$: 
    \begin{equation*}
        G(z_{(r)}) = [G_1(z_{(r),1}), G_2(z_{(r),2}), \cdots, G_{p_{(r)}}(z_{(r),{p_{(r)}}})]^T,
    \end{equation*}
    with
    \begin{equation*}
        G_j(z_{(r),j}) = \int^{z_{(r),j}}_{z'_{(r),j} = -\infty}  \mathcal{N}(z_{(r),j}|0,1)dz'_{(r),j}.
    \end{equation*}
    Then 
    \begin{equation*}
        dG(z_{(r)}) = p(z_{(r)})dz_{(r)} = \mathcal{N}(z_{(r)}|0,I)dz_{(r)}.
    \end{equation*}
    Let the decoder be 
    \begin{equation*}
        f_{\mu_x}(z_{(r)};\theta^*_t) = F^{-1} \circ G(z_{(r)}),
    \end{equation*}
    with
    \begin{equation*}
        \gamma^*_t = \frac{1}{t}.
    \end{equation*}
    Then we have that 
    \begin{align*}
        p_{\theta^*_t}(x_{(r)}) &= \int_{\mathbb{R}^{p_{(r)}}} p_{\theta^*_t}(x_{(r)}|z_{(r)})p(z_{(r)})dz_{(r)} \\&= \int_{\mathbb{R}^{p_{(r)}}} \mathcal{N}(x_{(r)}|F^{-1} \circ G(z_{(r)}), \gamma^*_tI)dG(z_{(r)}).
    \end{align*}
    Let $\xi = G(z_{(r)})$ such that 
    \begin{equation*}
        p_{\theta^*_t}(x_{(r)})= \int_{[0,1]^{p_{(r)}}} \mathcal{N}(x_{(r)}|F^{-1} (\xi), \gamma^*_tI)d\xi,
    \end{equation*}
    and let $x'_{(r)} = F^{-1} (\xi)$ such that $d \xi = dF(x'_{(r)}) = p_{gt}(x'_{(r)})d x'_{(r)}$. Plugging this expression into the previous $p_{\theta^*_t}(x_{(r)})$ we get
    \begin{equation*}
        p_{\theta^*_t}(x_{(r)})= \int_{\mathbb{R}^{p_{(r)}}} \mathcal{N}(x_{(r)}|x'_{(r)}, \gamma^*_tI)p_{gt}(x'_{(r)})d x'_{(r)}.
    \end{equation*}
    As $t \rightarrow \infty, \gamma^*_t$ becomes infinitely small and $\mathcal{N}(x_{(r)}|x'_{(r)}, \gamma^*_tI)$ becomes a Dirac-delta function, resulting in 
    \begin{equation*}
        \lim_{t \rightarrow \infty} p_{\theta^*_t}(x_{(r)}) = \int_{\mathcal{X}_{(r)}}\delta(x'_{(r)}-x_{(r)})p_{gt}(x'_{(r)})dx'_{(r)} = p_{gt}(x_{(r)}).
    \end{equation*}
    where $\mathcal{X}_{(r)}$ denotes the part of the data space corresponding to the observed dimensions of the data.
    
    Combining this result with Lemma \ref{lemma_r} that 
   there exists a decoder with parameter $\psi^*$ that recovers the ground truth  $p_{\psi^*}(r|x_{(r)}) = p_{gt}(r|x_{(r)}) $, we have that there is a sequence of decoders for the observed joint distribution such that
    \begin{equation*}
        \lim_{t \rightarrow \infty} p_{\theta^*_t, \psi^*_t} (x_{(r)}, r) = p_{gt}(x_{(r)}, r).
    \end{equation*}
    By continuity, we have 
    \begin{equation*}
        \lim_{t \rightarrow \infty} \log p_{\theta^*_t, \psi^*_t} (x_{(r)}, r) = \log p_{gt}(x_{(r)}, r).
    \end{equation*}
    According to Theorem \ref{LK}, if $w$ is bounded, then%
\begin{equation*}
    \lim_{K \rightarrow \infty} \mathcal{L}_K (\theta^*_{t}, \psi^*_{t}, \phi^*_{t}, \lambda^*_{t})  = \log p_{\theta^*_t, \psi^*_t}(x_{(r)}, r),\quad \forall \phi^*_t \in \Omega_\phi, \lambda^*_t \in \Omega_\lambda.
\end{equation*}
So we have
    \begin{equation*}
    \lim_{t\to\infty} \lim_{K \rightarrow \infty} \mathcal{L}_K (\theta^*_{t}, \psi^*_{t}, \phi^*_{t}, \lambda^*_{t})  = \log p_{gt}(x_{(r)}, r).
\end{equation*}
\end{proof}
\section{Parameter Settings for Simulation Study}
We followed the parameter settings in \cite{malinsky2021semiparametric} for the simulation study for mean estimation. For K = 3 missing variables, missingness patterns $(1, 0, 0)$, $(0, 1, 0)$, $(1, 1, 0)$, $(0, 0, 1)$, $(1, 0, 1)$, $(0, 1, 1)$, $(1, 1, 1)$, and $(0, 0, 0)$ were sampled with probabilities $0.169$, $0.153$, $0.136$, $0.119$, $0.102$,
$0.085$, $0.169$, and $0.068$ respectively (rounded to the $3^\text{rd}$ decimal place). Set the mean parameter vector for each missingness pattern $R=r$ is $\mu_0(r) = \Sigma_0 h(r)$ where 
\begin{align*}
h(1, 0, 0) &= (1.4, 1.6, 0.9)'\\
h(0, 1, 0) &= (1.9, 1.1, 1.4)'\\
h(1, 1, 0) &= (1.9, 1.6, 0.2)'\\
h(0, 0, 1) &= (0.5, 1.9, 2.1)'\\
h(1, 0, 1) &= (0.5, 2.4, 0.9)'\\
h(0, 1, 1) &= (1.0, 1.9, 1.4)'\\
h(1, 1, 1) &= (1.0, 2.4, 0.2)'\\
h(0, 0, 0) &= (1.4, 1.1, 2.1)'.
\end{align*}
and the variance-covariance matrix as 
\begin{equation*}
    \Sigma_0 = \begin{pmatrix}
4.4 & 1.3 & -2.8\\
1.3 & 3.2 & 1.3\\
-2.8 & 1.3 &3.5
\end{pmatrix},
\end{equation*}

It is straightforward to verify that under these parameter settings, the assumption of no self-censoring is met in the generated data.

\section{More Details on the Choice of the Latent Dimension $\kappa_2$}\label{app_dimzt}
\begin{figure}[tbh!]
    \centering
    \includegraphics[width = 0.8\textwidth]{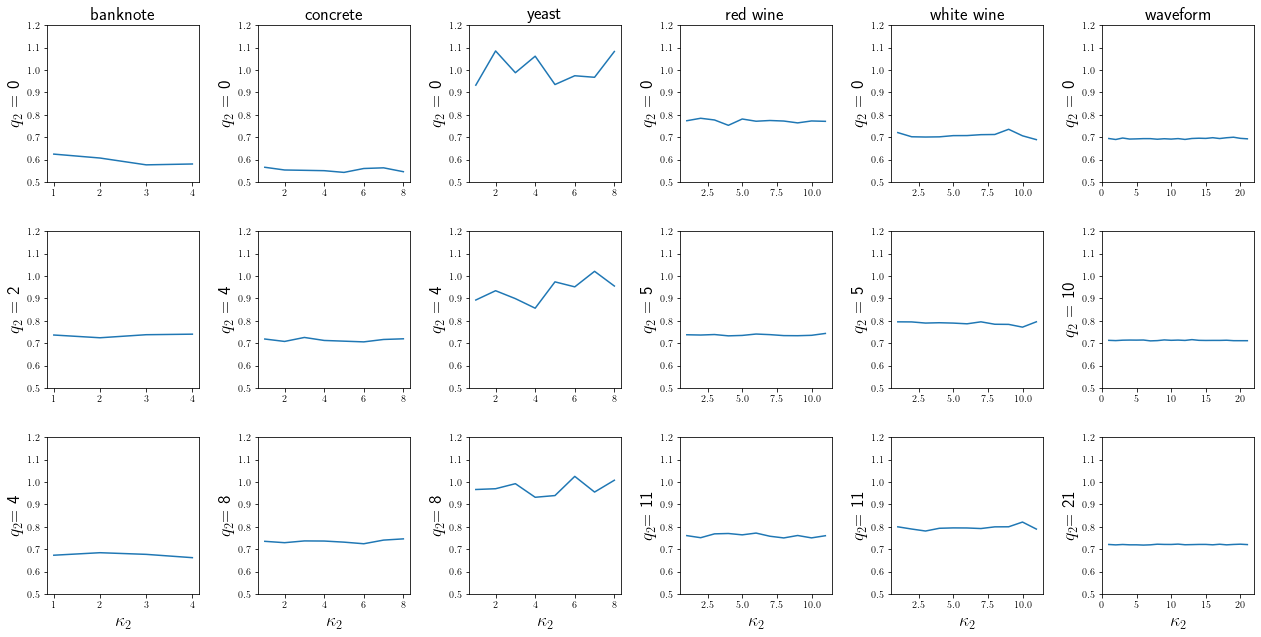}
    \caption{Imputation RMSE with different latent dimensions for the missingness mechanism, dim($\tilde{Z}$) in UCI datasets with simulated no self-censoring missingness. Each column is a different dataset, and each row represents the true latent dimension for the missingness mechanism in the simulation in the corresponding dataset. The y-axis shows the imputation RMSE on the synthetically missing entries. The x-axis shows the latent dimension for the missingness mechanism used in the learning model.}
    \label{fig:uciztplot}
\end{figure}
We experimented with the choice of latent dimensions for the missingness mechanism using the UCI datasets.  We simulated missingness on the original complete UCI data according to no self-censoring with the latent $q_2$ dimensions chosen from $0, \lfloor p/2\rfloor$ and $p$, where $p$ is the data dimension of each dataset. In each of these simulation settings, (fixing the latent dimensions for the data, which have been known for the complete datasets through cross-validation using the generative models,) we experimented with different latent dimensions for the missingness mechanism, $\kappa_2$ in the learning model. See Figure \ref{fig:uciztplot}. There does not appear to be significant difference among different choices of $\kappa_2$ in the learning model in each case, and choosing the dimension $\kappa_2$ to be $1$ works sufficiently well in the learning model from these experiments, no matter what the true latent dimensions are. This empirical result is consistent with the theoretical conclusion in Lemma \ref{lemma_r}.

\section{Additional Experimental Results}
\subsection{Maximum Mean Discrepancy (MMD) for synthetic data experiments}
We present more numerical results on the comparison of generated data in the synthetic data experiments with the condition of no self-censoring in Section 5.1 in terms of maximum mean discrepancy (MMD). As shown in Table \ref{mmdtb} below, the proposed method shows a clear advantage over the other state-of-the-art methods in terms of mean of MMD when comparing the generated data distribution with the ground-truth complete-data distribution $p_{gt}(x)$.
\begin{table}[!htb]
\caption{Mean(SD) for maximum mean discrepancy (MMD) of simulated 3D no self-censoring missing data using different imputation methods, for both linear and nonlinear censoring with corresponding generative model architectures.}
\label{mmdtb}
\begin{center}
\begin{tabular}{lcc|cc}
\hline
\multirow{2}{*}{Method} & \multicolumn{2}{c|}{Using latent variables}& \multicolumn{2}{c}{No latent variables}\\
& linear & nonlinear& linear & nonlinear \\ \hline
Proposed &  \B 0.0086 (0.0074) &  \B 0.0062 (0.0071)& \B 0.0126 (0.0203) &\B 0.0053 (0.0039)\\
GINA &  0.0098 (0.0102) & 0.0073  (0.0069) & 0.0143 (0.0131) & 0.0061 (0.0039)\\
Not-MIWAE & 0.0167 (0.0192)  & 0.0085 (0.0083)& 0.0145 (0.0155) & 0.0059 (0.0038)\\
\hline
\end{tabular}%
\end{center}
\end{table}
\subsection{Experiments with no self-censoring in UCI data}
For the real-world data from UCI machine learning repository, we evaluated the imputation performance in scenarios where the missingness mechanism aligns with the no self-censoring assumption, but without using latent variables for the missingness mechanism. To emphasize the role of the MNAR mechanism, we primarily focused on linear transformations for generating the missingness indicators from the data.

We simulated the missingness according to no self-censoring without latent variables as described in the synthetic experiments in Section 5.1. 

The results, as presented in Table \ref{uci_nsc}, consistently indicated that our method outperformed alternative approaches in the missingness mechanism of no self-censoring, even though the simulation of missingness did not rely on latent variables. 

\begin{table}[tbh!]
\begin{center}
\caption{Mean(SD) for imputation RMSE on  UCI Data with no self-censoring. By default, the missingness mechanisms in deep generative models are general linear structures, while (nl) indicates that a nonlinear structure is used instead, and (s) means self-censoring structure which particularly applies to not-MIWAE.}
\label{uci_nsc}
\begin{tabular}{lrrrrrr}
\hline
Method & \multicolumn{1}{c}{banknote} &   \multicolumn{1}{c}{concrete}& \multicolumn{1}{c}{yeast} &\multicolumn{1}{c}{red wine} & 
\multicolumn{1}{c}{white wine} & \multicolumn{1}{c}{waveform}\\ \hline
Proposed &  \B 0.62(0.09)   & \B 0.65(0.05) & \B 0.91(0.08) & \B 0.71(0.03) & \B 0.75(0.02) & \B 0.69(0.01)\\
Proposed(nl) &\B 0.62(0.08)  & \B 0.67(0.05) & \B 0.95(0.10) & \B 0.73(0.04)& \B 0.77(0.03) & \B 0.70(0.01)\\
GINA &  0.82(0.12)  & 0.85(0.04) & 1.06(0.05) & 0.87(0.04) &  0.91(0.02) &0.78(0.01)\\
GINA(nl) & 0.85(0.12)  & 0.90(0.05) & 1.11(0.05) & 0.91(0.04) & 0.96(0.03) & 0.79(0.01)\\
Not-MIWAE & 0.76(0.15)   & 0.81(0.10) & 1.16(0.18) & 0.85(0.07) & 0.88(0.08) & 0.73(0.01)\\
Not-MIWAE(nl) & 0.92(0.11)  & 0.85(0.08) & 1.06(0.06) & 0.86(0.04) & 0.90(0.04) & 0.77(0.02)\\
Not-MIWAE(s) & 1.82(0.78)  & 1.74(0.55) & 2.44(0.67) & 2.27(0.63) & 2.37(0.4) & 1.94(0.58)\\
MICE & 1.00(0.13)  & 1.01(0.05) & 4.79(2.17) & 1.01(0.05) & 1.02(0.03) & 0.94(0.02)\\
MissForest &0.87(0.15)  & 0.94(0.08) & 1.46(0.29) & 0.96(0.06) & 0.97(0.07) & 0.83(0.01)\\
\hline
\end{tabular}%
\end{center}
\vspace{-20pt}
\end{table}

\subsection{Experiments with self-censoring in UCI data}
Despite its mild and flexible nature as a condition for nonparametric identification in MNAR data, it is important to acknowledge that the no self-censoring model assumption may not always hold in real MNAR data scenarios. To assess the robustness and adaptability of our model in the presence of model misspecification, we conducted a comparative analysis by contrasting the results of different methods under the opposite of no self-censoring within the MNAR regime --- the self-censoring model. 

In the self-censoring model, we followed the methodology outlined in \citet{ipsen2020not}. Here, we retained the values for half of the features, while for the other half, their values were set to be missing if they exceeded the average value of that respective feature. Among all MNAR mechanisms, this particular setting stands out as one of the most challenging scenarios for a method based on the no self-censoring model, as it specifically encompasses the self-censoring aspect that our model inherently lacks.

In the results presented in Table \ref{robust}, it is observed that not-MIWAE with a linear self-censoring model for missingness consistently outperforms other models. This aligns with expectations since it is the only model that matches the self-censoring assumption. However, when the self-censoring missingness model is replaced by a general linear one, as seen in both not-MIWAE and GINA, there is a noticeable decline in performance, despite the linear self-censoring being a special case of the linear model.

Remarkably, our proposed model, which does not incorporate any modeling of self-censoring, achieves relatively acceptable results even in the presence of model misspecification, particularly when using a nonlinear no self-censoring missingness model.
This observation implies that the correlation between data variables could make the no self-censoring model a practical assumption, even in scenarios where self-censoring might be a factor to consider. 
\begin{table}[tbh!]
\begin{center}
\caption{Mean(SD) for imputation RMSE on  UCI Data with self-censoring. The best two results are highlighted. By default, the missingness mechanisms in deep generative models are general linear structures, while (nl) indicates that a nonlinear structure is used instead, and (s) means a self-censoring structure which particularly applies to not-MIWAE. The results show a certain degree of robustness of the proposed method in the case of model misspecification.}
\label{robust}
\begin{tabular}{lrrrrrr}
\hline
Method & \multicolumn{1}{c}{banknote}  & \multicolumn{1}{c}{concrete}& \multicolumn{1}{c}{yeast} &  \multicolumn{1}{c}{red wine} & 
\multicolumn{1}{c}{white wine} & \multicolumn{1}{c}{waveform}\\ \hline
Proposed &  1.38(0.08)  & 1.77(0.10) & 1.65(0.05) & 1.52(0.03) & 1.59(0.02) & 1.45(0.02)\\
Proposed(nl) & \B 1.21(0.04) &   \B 1.67(0.03) & 1.65(0.05) &\B 1.44(0.02) & 1.57(0.02) & 1.43(0.01)\\
GINA & 2.66(1.46)  & 2.37(1.32) & 1.52(0.05) & 1.88(0.56) & 1.74(0.26) & 1.55(0.10)\\
Not-MIWAE  &1.73(0.61)  & 1.99(0.67) & \B 1.51(0.11) & 1.79(0.65) & 1.63(0.01) & 1.52(0.04) \\
Not-MIWAE(s)&  \B 1.04(0.16) & \B 1.21(0.10) & \B 1.27(0.12) &  \B 1.15(0.01) & \B 1.20(0.03) & \B 0.94(0.02)\\
MICE &  1.41(0.00)  &1.87(0.00) & 1.76(0.00) & 1.68(0.00) & \B 1.41(0.00) & \B 1.42(0.00)\\
MissForest & 1.28(0.01)  & 1.81(0.02) & 1.72(0.01) & 1.62(0.02) & 1.63(0.01) & 1.51(0.00)\\
\hline
\end{tabular}%
\end{center}
\vspace{-20pt}
\end{table}

\section{Implementation Details in Experiments}
We consider the implementation for both linear missingness models and nonlinear missingness models, that is, whether $p_{\psi}(r|x)$ is a linear or nonlinear function of $x$.
For learning a linear missingness model, there are no hidden layers from input $X$ and $\tilde{Z}$ to the output $R$. We can remove the connection between $X_j$ and $R_j$ for all $1 \leq j \leq p$ from the fully connected neural network structure to arrive at a structure of no self-censoring. To model a nonlinear relationship, at least one hidden layer is added with a nonlinear transformation such as RELU or $tanh$. There is no direct way of removing edges from such a network to obtain a no self-censoring effect. One brute-force way to implement this nonlinear self-censoring model is to use a whole set of unique edges for each missingness indicator $R_j$ for $1 \leq j \leq p$ and iterate through the dimensions. \textcolor{black}{Instead,} we propose to use a shortcut by remasking each corresponding data variable $X_j$ with zero for each missingness indicator $R_j$, while still sharing all the weight parameters for each dimension $j$ for $1 \leq j \leq p$. This is inspired by the practice of initial zero imputation prior to training used for DLVMs for missing data, which deactivates the propagation of the loss on the associated weight parameters \citep{nazabal2020handling}. The number of parameters of the model can be kept down in this regime, and even smaller than a general missingness mechanism model without considering no self-censoring.

For all experiments, TensorFlow probability \citep{dillon2017tensorflow} and the Adam optimizer \citep{kingma2014adam} were employed, with the learning rate set at 0.001. No regularization techniques were applied. For MICE and missForest, we utilize the implementation in the Python package \code{sklearn} directly.
In the following, we use $D$ for the dimension of the latent space for the data, and $H$ for the number of neurons in the hidden layer. A 1-dimensional latent space was employed for the missingness indicators. In the inference network, we mainly utilized a structure of $p$-$H$-$H$-$(D+1)$, which maps the data (with missing data filled with zero) into distributional parameters for the two latent spaces. The decoder $p_{\theta}(x|z)$ followed a structure of $D$-$H$-$H$-$p$. In the linear implementation of the missingness model, there was no hidden layer, while in the nonlinear one, we included one hidden layer. All neural networks incorporated $tanh$ activations, except for the output layer where no activation function was used. The generative models adopted the objectives with $K=20$ importance samples. The observational noise $\gamma$ for the data variables was set as a learnable parameter in the model. All methods were trained with a batch size of 16 for 10k epochs. Imputation RMSE was estimated using 10k importance samples. For data generation, we generated from the latent variables one sample for each data dimension in each iteration.
\paragraph{Simulation}  We used latent dimension $H=128$ with $D=3$ for low-dimensional data, and $D = 10$ for high-dimensional data in the simulation. In the mixture of Gaussians, we simplified the network structures to have one hidden layer for both the encoder and decoder for the data, $p$-$H$-$(D+1)$ and $D$-$H$-$p$ respectively, and the maximum number of epochs were set to be 20k.
\paragraph{UCI data} We used latent dimension $D = p-1$ for UCI data, where the data were standardized to have a mean of 0 and unit variance before introducing missing values.  And the data are randomly shuffled at the beginning of each iteration of the experiment.
\paragraph{HIV data} As the data is binary, one can not differentiate between the zero values and the initial zero imputation, we input both the data and missingness indicators to the encoder and employed an encoder structure similar to the one in the synthetic experiment for a mixture of Gaussians. Additionally, a Bernoulli decoder instead of a Gaussian decoder was used for the data.
\paragraph{Yahoo! data} Prior to training, all user ratings are normalized to fall within the range of 0 to 1 (this normalization will be reversed during evaluation). Gaussian likelihood with a variance of $\gamma = 0.02$ is employed for all the generative models. A 20-dimensional latent space for the data is utilized, and the decoder follows a structure of 20-10-$p$. The inference network adopts the point net structure proposed by \cite{ma2018eddi}, employing a 20-dimensional feature mapping $h$ parameterized by a single-layer neural network and 20-dimensional ID vectors for each variable. The symmetric operator is defined as the summation operator. The missing model 
is parameterized by a linear neural network. All methods underwent training for 100 epochs with a batch size of 100. Imputation RMSE was estimated using 10 importance samples.









\vskip 0.2in
\bibliography{sample}

\end{document}